\documentclass{llncs}
\usepackage{amssymb,boxedminipage,amsmath}
\usepackage[export]{adjustbox}
\usepackage{tikz}
\newcommand{\NP}{{\sf NP}}
\newcommand{\PP}{{\sf P}}

\newcommand{\W}{{\sf W}}

\newcommand{\ssi}{\subseteq_i}

\usepackage{comment}

\newcommand{\vd}{{\sf vd}}

\newcommand{\ec}{{\sf ec}}

\title{Contraction and Deletion Blockers\\ for Perfect Graphs and $H$-free Graphs\thanks{A number of results in this paper have appeared in extended abstracts of the proceedings of 
CIAC~2016~\cite{DPPR15}, ISCO 2016~\cite{PPR16}, LAGOS 2017~\cite{PPR17b} and TAMC 2017~\cite{PPR17}.}}
\author{\"Oznur Ya\c{s}ar Diner\inst{1}\thanks{Author supported by Marie Curie International Reintegration Grant PIRG07/GA/2010/268322.}\and
Dani\"el Paulusma\inst{2}\thanks{Author supported by EPSRC (EP/K025090/1) and the Leverhulme Trust (RPG-2016-258).}
\and\\ Christophe Picouleau\inst{3} \and Bernard Ries\inst{4}
}

\institute{Kadir Has University, Istanbul, Turkey
\texttt{oznur.yasar@khas.edu.tr}
\and
Durham University, Durham, UK, \texttt{daniel.paulusma@durham.ac.uk}
\and
CNAM, Laboratoire CEDRIC, Paris, France
\texttt{christophe.picouleau@cnam.fr}
\and
University of Fribourg, 
Fribourg, Switzerland,  \texttt{bernard.ries@unifr.ch}
}

\begin{document}
\maketitle
\setcounter{footnote}{0}

\begin{abstract}
We study the following problem: for given integers $d$, $k$ and graph~$G$, can we reduce some fixed graph parameter~$\pi$ of $G$ by at least $d$ via at most $k$ graph operations from some fixed set $S$? As parameters we take the chromatic number~$\chi$, clique number $\omega$ and independence number $\alpha$, and as operations we choose the edge contraction \ec\  and vertex deletion \vd. 
We determine the complexity of this problem for $S=\{\ec\}$ and $S=\{\vd\}$ and $\pi\in \{\chi,\omega,\alpha\}$ for a number of subclasses of perfect graphs.
We use these results to determine the complexity of the problem for $S=\{\ec\}$ and $S=\{\vd\}$ and $\pi\in \{\chi,\omega,\alpha\}$ 
restricted to $H$-free graphs.
\end{abstract}

\section{Introduction}
\label{s-intro}

A typical graph modification problem aims to modify a graph $G$, via a small number of operations from a specified set~$S$, into some other graph $H$ that has a certain desired property, which usually describes a certain graph class~{${\cal G}$ to which $H$ must belong. 
In this way a variety of classical graph-theoretic problems is captured.
For instance, if only $k$ vertex deletions are allowed and $H$ must be an independent set or a clique, we obtain the {\sc Independent Set} or {\sc Clique} problem, respectively. 

Now, instead of fixing a particular graph class~${\cal G}$, we {\it fix a certain graph parameter~$\pi$}. That is, for a fixed set~$S$ of graph operations, 
we ask, given a graph $G$, integers~$k$ and~$d$, whether $G$ can be transformed into a graph $G'$ by using at most $k$ operations from $S$, such that $\pi(G')\leq \pi(G)-d$. 
The integer~$d$ is called the {\it threshold}.
Such problems are called {\it blocker problems}, as the set of vertices or edges involved ``block''  some desirable graph property, such as being colourable with only a few colours. Identifying the part of the graph responsible for a significant decrease of the parameter under consideration gives crucial information on the graph. 

Blocker problems have been given much attention over the 
last few years, see for instance~\cite{BBPR,BTT11,Bentz,CWP11,PBP14,PBP15,RBPDCZ10,T10}.
Graph parameters considered were the chromatic number, the independence number, the clique number, the matching number, the
weight of a minimum dominating set 
and the vertex cover number. So far, the set $S$ always consisted of a single graph operation, which was a vertex deletion, edge deletion or an edge addition. 
In this paper, we keep the restriction on the size of $S$ by letting $S$ consist of either a single vertex deletion or, for the first time, a single edge contraction. 
As graph parameters we consider the independence number~$\alpha$, the clique number~$\omega$ and the chromatic number~$\chi$.

Before we can define our problems formally, we first need to give some terminology.
The \emph{contraction} of an edge $uv$ of a graph $G$ removes the vertices $u$ and $v$ from $G$, and replaces them by a new vertex made adjacent to precisely those vertices that were adjacent to $u$ or $v$ in $G$ (neither introducing self-loops nor multiple edges). We say that $G$ can be \emph{$k$-contracted} or \emph{$k$-vertex-deleted} into a graph~$G'$, if $G$ can be modified into $G'$ by a sequence of at most~$k$ edge contractions or vertex deletions, respectively. We let $\pi$ denote the (fixed) graph parameter; as mentioned, in this paper $\pi$ belongs to $\{\alpha,\omega,\chi\}$. 

We are now ready to define our decision problems in a general way:
\begin{center}
\begin{boxedminipage}{.99\textwidth}
\textsc{\sc Contraction Blocker($\pi$)}\\[2pt]
\begin{tabular}{ r p{0.8\textwidth}}
\textit{~~~~Instance:} &a graph $G$ and two integers $d,k\geq 0$\\
\textit{Question:} &can $G$ be $k$-contracted into a graph $G'$ such that $\pi(G')\leq \pi(G)-d$?
\end{tabular}
\end{boxedminipage}
\end{center}

\begin{center}
\begin{boxedminipage}{.99\textwidth}
\textsc{\sc Deletion Blocker($\pi$)}\\[2pt]
\begin{tabular}{ r p{0.8\textwidth}}
\textit{~~~~Instance:} &a graph $G$ and two integers $d,k\geq 0$\\
\textit{Question:} &can $G$ be $k$-vertex-deleted into a graph $G'$ such that~$\pi(G')\leq \pi(G)-d$?
\end{tabular}
\end{boxedminipage}
\end{center}

\medskip
\noindent
If we remove $d$ from the input and fix it instead, then we call the resulting problems {\sc $d$-Contraction Blocker($\pi$)} and 
{\sc $d$-Deletion Blocker($\pi$)}, respectively. 

\begin{center}
\begin{boxedminipage}{.99\textwidth}
\textsc{\sc $d$-Contraction Blocker($\pi$)}\\[2pt]
\begin{tabular}{ r p{0.8\textwidth}}
\textit{~~~~Instance:} &a graph $G$ and an integer $k\geq 0$\\
\textit{Question:} &can $G$ be $k$-contracted into a graph $G'$ such that $\pi(G')\leq \pi(G)-d$?
\end{tabular}
\end{boxedminipage}
\end{center}

\begin{center}
\begin{boxedminipage}{.99\textwidth}
\textsc{\sc $d$-Deletion Blocker($\pi$)}\\[2pt]
\begin{tabular}{ r p{0.8\textwidth}}
\textit{~~~~Instance:} &a graph $G$ and an integer $k\geq 0$\\
\textit{Question:} &can $G$ be $k$-vertex-deleted into a graph $G'$ such that~$\pi(G')\leq \pi(G)-d$?
\end{tabular}
\end{boxedminipage}
\end{center}

The goal of our paper is to increase our understanding of the complexities of {\sc Contraction Blocker($\pi$)} and {\sc Deletion Blocker($\pi$)} fo
$\pi \in \{\omega,\chi,\alpha\}$. In order to do so, we will also consider the problems {\sc $d$-Contraction Blocker($\pi$)} and {\sc $d$-deletion blocker($\pi$)}.

\subsection{Known Results and Relations to Other Problems}\label{s-known}

It is known that {\sc Deletion Blocker($\alpha$)} is polynomial-time solvable for bipartite graphs, as proven both by
Bazgan, Toubaline and Tuza~\cite{BTT11} and Costa, de Werra and Picouleau~\cite{CWP11}.
The former authors also proved that  {\sc Deletion Blocker($\alpha$)} is polynomial-time solvable for cographs and graphs of bounded treewidth. The latter authors also proved that
for $\pi \in \{\omega,\chi\}$, {\sc Deletion Blocker($\pi$)} is polynomial-time solvable for cobipartite graphs. Moreover, they showed that for $\pi \in \{\omega,\chi,\alpha\}$, {\sc Deletion Blocker($\pi$)} is \NP-complete for the class of split graphs, but becomes polynomial-time solvable for this graph class if $d$ is fixed.

By using a number of example problems we will now illustrate how the blocker problems studied in this paper relate to a number of other problems known in the literature.
As we will see, this immediately leads to new complexity results for the blocker problems.

\medskip
\noindent
{\bf 1. Hadwidger Number and Club Contraction.}
The \sc Contraction Blocker($\alpha$)} problem generalizes the well-known {\sc Hadwiger Number} problem, which is that of testing whether a graph can be contracted into the complete graph $K_r$ on $r$ vertices 
for some given integer $r$. Indeed, we obtain the latter problem from the first by restricting instances to instances $(G,d,k)$ where $d=\alpha(G)-1$ and $k=|V(G)|-r$. Note that the diameter and independence number of $K_r$ are both equal to~1. Hence, one can also generalize {\sc Hadwiger Number} in another way: the {\sc Club Contraction} problem (see e.g.~\cite{GHHP14}) is that of testing whether a graph $G$ can be $k$-contracted into a graph with diameter at most~$s$ for some given integers~$k$ and $s$. 
As such, {\sc Contraction Blocker($\alpha$)} can be seen as a natural counterpart of {\sc Club Contraction}.

\medskip
\noindent
{\bf 2. Graph transversals.}
Blocker problems generalize so-called graph transversal problems.
To explain the latter type of problems, for a family of graphs~${\cal H}$, the {\sc ${\cal H}$-transversal} problem is to test if a graph $G$ can be $k$-vertex-deleted, for some integer~$k$, into a graph~$G'$ that has no induced subgraph isomorphic to a graph in~${\cal H}$. For instance, the problem {\sc $\{K_2\}$-transversal} is the same as {\sc Vertex Cover}. Here are some examples of specific connections between graph transversals and blocker problems.

\begin{itemize}
\item Let ${\cal H}$ be the family $\{K_p\; |\; p\geq 2\}$ of all complete graphs on at least two vertices. Then {\sc ${\cal H}$-transversal} is equivalent to {\sc Deletion Blocker($\omega$)} restricted to instances $(G,d,k)$ with $d=\omega(G)-1$.
\item In our paper we will prove that for a graph $G$ with at least one edge and an integer~$k \geq 1$, the instance $(G,\omega(G)-1,k)$ is a yes-instance of {\sc Deletion Blocker($\omega$)} if and only if $(G,k)$ is a yes-instance of {\sc Vertex Cover}.
\item
The {\sc Odd Cycle Transversal} problem is to test whether a given graph can be made bipartite by removing at most $k$ vertices for some given integer $k\geq 0$. This problem is \NP-complete~\cite{LY80}, and it is equivalent to {\sc Deletion Blocker($\chi$)} for instances $(G,d,k)$ where $d=\chi(G)-2$.
\item  
The {\sc ${d}$-Transversal} or {\sc $d$-Cover} problem~\cite{CWP11} is to decide whether a graph $G=(V,E)$ contains a set $V'$ that intersects each 
maximum set satisfying some specified property~$\pi$ by at least $d$ vertices. For instance, if the property is being an independent set, {\sc $1$-Transversal} is equivalent to {\sc 1-Deletion Blocker($\alpha$)}. 
\end{itemize}

\noindent
{\bf 3. Bipartite Contraction.}
The problem {\sc Bipartite Contraction} is to test whether a graph can be made bipartite by at most $k$ edge contractions.
Heggernes et al.~\cite{HHLP} proved that this problem is \NP-complete.
It is readily seen that {\sc $1$-Contraction Blocker($\chi$)}   and {\sc Bipartite Contraction} are equivalent for graphs of chromatic number~3.

\medskip
\noindent
{\bf 4. Maximum induced bipartite subgraphs.}
The {\sc Maximum Induced Bipartite Subgraph} problem is to decide if a given graph contains an induced bipartite subgraph with at least $k$ vertices for some integer~$k$.
Addario-Berry et al.\cite{AKKLR10} proved that this problem is  \NP-complete for the class of 3-colourable perfect graphs. 
We observe that, for 3-colourable graphs, {\sc $1$-Deletion Blocker($\chi$)} is equivalent to {\sc Maximum Induced Bipartite Subgraph}.

\medskip
\noindent
{\bf 5. Cores.}
The two problems {\sc $1$-Deletion Blocker($\alpha$)}  and {\sc $1$-Deletion Blocker($\omega$)} are equivalent to testing whether the input graph contains a set of $S$ of size at most~$k$ that intersects every maximum independent set or  every maximum clique, respectively. If $k=1$, these two problems become  equivalent to testing whether the input graph contains a vertex that is in every maximum independent set, or in every maximum clique, respectively. In particular, the intersection of all maximum independent sets is known as the  {\it core} of a graph. 
Properties of the core have been well studied (see, for example,~\cite{HHS82,LM02,LM12}).
In particular, Boros, Golumbic and Levit~\cite{BGL02} proved that computing if the core of a graph has size at least $\ell$ is co-\NP-hard for every fixed $\ell\geq 1$. Taking $\ell=1$ gives co-\NP-hardness of  {\sc $1$-Deletion Blocker($\alpha$)}.

\medskip
\noindent
{\bf 6. Critical vertices and edges.}
The restriction $d=k=1$ has also been studied when $\pi=\chi$.
A vertex of a graph~$G$ is {\it critical} if its deletion reduces the chromatic number of~$G$ by~1. An edge of a graph is
{\it critical} or {\it contraction-critical} if its deletion or contraction, respectively, reduces the chromatic number of~$G$ by~1.
The problems {\sc Critical Vertex}, {\sc Critical Edge} and {\sc Contraction-Critical} are to test if a graph has a critical vertex, critical edge or contraction-critical edge, respectively. We note that the first two problems are the restrictions of {\sc Contraction Blocker($\chi$)} and {\sc Deletion Blocker($\chi$)} to instances $(G,d,k)$ where $d=k=1$. Complexity dichotomies exist for each of the three problems on $H$-free graphs, and moreover the latter two problems are shown to be equivalent~\cite{PPR17b}.
Graphs with a critical (or equivalently contraction-critical) edge are also called {\it colour-critical} (see, for instance,~\cite{PY17}).

\medskip
\noindent
Due to links to problems as the ones above, it is of no surprise that many results for blocker problems are known implicitly in the literature already in various settings. 
For example, Belmonte et al.~\cite{BGHP13} proved that  $1$-{\sc Contraction Blocker($\Delta$)}, where $\Delta$ denotes the maximum vertex-degree, is \NP-complete
even for split graphs. We make use of several known complexity results for some of the related problems stated above for proving our results.

\subsection{Our Results}

In Section~\ref{s-known} we mentioned that {\sc Deletion Blocker($\pi$)} is known to be \NP-complete for $\pi\in \{\alpha,\omega,\chi\}$ even when restricted to special graph classes.  
Non-surprisingly, {\sc Contraction Blocker($\pi$)} is \NP-complete for $\pi\in \{\alpha,\omega,\chi\}$ as well (this follows from our results in Section~\ref{sec:general}, but it is also easy to show this directly).

Due to the above,
it is natural to restrict inputs to some special graph classes in order to obtain tractable results and to increase our understanding of the computational hardness of the problems. 
Note that it is not always clear whether {\sc Contraction Blocker($\pi$)} and {\sc Deletion Blocker($\pi$)} belong to \NP\ when restricted to a graph class~${\cal G}$. However, when ${\cal G}$ is closed under edge contraction or vertex deletion, respectively, and $\pi$ can be verified in polynomial time, then membership of \NP\ holds: we can take as certificate the sequence of edge contractions or vertex deletions, respectively.

\begin{table}[htp]
\begin{center}
\hspace*{-0.5cm}
\begin{tabular}{|l||c|c||c|c|}
\hline
&\multicolumn{2}{c||}{{\sc \textsc{Contraction Blocker($\pi$)}}}&\multicolumn{2}{c|}{{\sc \textsc{Deletion Blocker}($\pi$)}}\\
\hline
\textbf{Class}&$\pi=\alpha$&$\pi=\omega=\chi$&$\pi=\alpha$&$\pi\in \{\omega,\chi\}$\\
\hline
\hline
tree& \PP & \PP& \hspace*{9.2mm}\PP~\cite{BTT11,CWP11}& \PP\\
\hline
bipartite  ($\chi=2$)&\NP-h& \PP & \hspace*{9.5mm}\PP~\cite{BTT11,CWP11}& \PP\\
\hline
cobipartite&$d=1$: \NP-c& \NP-c; $d$ fixed: \PP& \PP& \hspace*{4.7mm} \PP~\cite{CWP11}\\
\hline
cograph&\PP& \PP& \hspace*{3mm} \PP~\cite{BTT11}& \PP\\
\hline
split&\NP-c; $d$ fixed: \PP&\NP-c; $d$ fixed: \PP&\NP-c; $d$ fixed: \PP~\cite{CWP11}& \NP-c; $d$ fixed: 
\PP~\cite{CWP11}\\
\hline
interval&& \PP& &  \PP\\
\hline
chordal&\NP-c&$d=1$: \NP-c&\NP-c & $d=1$: \NP-c\\
\hline
$C_4$-free perfect \& $\omega=3$& & $d=1$: \NP-c& &  \\
\hline
perfect &$d=1$: \NP-h& $d=1$: \NP-h$ $& \NP-c &  $d=1$: \NP-c\\
\hline
\end{tabular}\vskip5mm
\caption{Summary of results for subclasses of perfect graphs. Here \NP-c  and \NP-h stand for \NP-complete and \NP-hard, respectively, whereas \PP\ stands for polynomial-time solvable. A blank entry indicates an open case. All entries apart from the five referenced ones and their consequences for chordal and perfect graphs are new results proven in Part~I of this paper.}\label{t-thetable2}
\end{center}
\end{table}

\vspace*{-1cm}
We present our results in two parts.

\medskip
\noindent
{\bf Part I.}
In the first part of our paper 
we focus on the class of perfect graphs and a number of well-known subclasses of perfect graphs. Most of these classes are not only closed under vertex deletion but also under edge contraction. This enables us to get unified results for the cases $\pi=\omega$ and $\pi=\chi$ (note that $\omega=\chi$ holds by definition of a perfect graph). Another reason for considering subclasses of perfect graphs is that $\alpha$, $\omega$, $\chi$ can be computed in polynomial time for perfect graphs; Gr\"otschel, Lov\'asz, and Schrijver~\cite{GLS84} proved this for $\chi$ and thus for $\omega$, whereas the result for $\alpha$ follows from combining this result with the fact that perfect graphs are closed under complementation. This helps us with finding tractable results or at least with obtaining membership of \NP\ (if in addition the subclass under consideration is closed under edge contraction or vertex deletion). 

Table~\ref{t-thetable2} gives an overview of the known results and our new results for the classes of perfect graphs we consider. We have unified results for the cases $\pi=\omega$ and $\pi=\chi$ even for the
perfect graph
classes in this table that are not closed under edge contraction, namely
the classes of bipartite graphs; 
$C_4$-free perfect graphs with clique number~3; and the class of perfect graphs itself. As the class of perfect graphs
is not closed under edge contraction we could for perfect graphs only deduce that the three contraction blocker problems are \NP-hard (even if $d=1$).
As the class of cographs coincides with the class of $P_4$-free graphs (where $P_r$ denotes the $r$-vertex path) and split graphs are $P_5$-free, the corresponding rows in Table~\ref{t-thetable2} show a complexity jump of all our problems for $P_t$-free graphs from $t=4$ to $t=5$. 
Recall also from Section~\ref{s-known} that the {\sc Hadwiger Number} problem is a special case of {\sc Contraction Blocker($\alpha$)})
As such, our polynomial-time result in Table~\ref{t-thetable2} for {\sc Contraction Blocker($\alpha$)} restricted to cographs generalizes a result of Golovach et al.~\cite{GHHP14}, who proved that the {\sc Hadwiger Number} problem is polynomial time solvable on cographs.\\

\noindent
{\bf Part II.}
In the second part of our paper we give several dichotomy results. First we give, for $\pi\in \{\alpha,\omega,\chi\}$, complete classifications of {\sc Deletion Blocker($\pi$)} and {\sc Contraction Blocker($\pi$)} depending on the size of $\pi$, that is, we prove the following theorem.

\begin{theorem}\label{thm:NPC}
The following six dichotomies hold:
\begin{itemize}
\item [{\bf (i)}] {\sc Contraction Blocker($\alpha$)}  is polynomial-time solvable for graphs with $\alpha=1$ and 
{\sc 1-Contraction Blocker($\alpha$)} is \NP-complete for graphs with $\alpha=2$;\\[-8pt]
\item[{\bf (ii)}] {\sc Contraction Blocker($\chi$)}  is polynomial-time solvable for graphs with $\chi=2$ and 
 {\sc $1$-Contraction Blocker($\chi$)} is \NP-complete for graphs with $\chi=3$;\\[-8pt]
\item  [{\bf (iii)}] {\sc Contraction Blocker($\omega$)} is polynomial-time solvable for graphs with $\omega=2$ and
{\sc $1$-Contraction Blocker($\omega$)} is \NP-complete for graphs with $\omega=3$;\\[-8pt]
\item  [{\bf (iv)}] {\sc Deletion Blocker($\alpha$)} is polynomial-time solvable for graphs with $\alpha=1$ and\\ 
{\sc $1$-Deletion Blocker($\alpha$)} is \NP-complete for graphs with $\alpha=2$;\\[-8pt]
\item [{\bf (v)}] {\sc Deletion Blocker($\chi$)} is polynomial-time solvable for graphs with $\chi=2$ and\\
{\sc $1$-Deletion Blocker($\chi$)} is \NP-complete for graphs with $\chi=3$;\\[-8pt]
\item [{\bf (vi)}] {\sc Deletion Blocker($\omega$)} is polynomial-time solvable for graphs with $\omega=1$ and\\
{\sc $1$-Deletion Blocker($\omega$)}  is \NP-complete for graphs with $\omega=2$;\end{itemize}
\end{theorem}

In particular we extend the hardness proof of Theorem~\ref{thm:NPC}~(iii) in order to obtain the hardness result  for $C_4$-free perfect graphs with $\omega=3$ in Table~\ref{t-thetable2}.
We note that some of the results in Table~\ref{t-thetable2}, such as this result, may at first sight seem somewhat arbitrary. However, we need the result for $C_4$-free perfect graphs with $\omega=3$ and other results of Table~\ref{t-thetable2} to prove our other results  of  the second part of our paper. Namely, by combining the results for subclasses of perfect graphs with other results, we obtain complexity dichotomies for our six blockers problems restricted to $H$-free graphs, that is, graphs that do not contain some (fixed) graph $H$ as an induced subgraph. These dichotomies are stated in the following summary; here, $P_r$ is the $r$-vertex path, $C_3$ is the triangle, and the paw is the triangle with an extra vertex adjacent to exactly one vertex of the triangle, whereas $\ssi$ denotes the induced subgraph relation and $\oplus$ denotes the disjoint union of two vertex disjoint graphs.

\begin{theorem}\label{t-mainmainmain}
Let $H$ be a graph. Then the following holds:
\begin{itemize}
\item [(i)] 
If $H\subseteq_i P_4$, then  {\sc Deletion Blocker($\alpha$)} is polynomial-time solvable for $H$-free graphs, otherwise it is 
\NP-hard or co-\NP-hard for $H$-free graphs.\\[-8pt]
\item [(ii)] 
If $H\subseteq_i P_4$, then {\sc Contraction Blocker($\alpha$)} is polynomial-time solvable for $H$-free graphs, otherwise it is
\NP-hard for $H$-free graphs.\\[-8pt]
\item [(iii)]
f $H\subseteq_i P_4$, then {\sc Deletion Blocker($\omega$)} is polynomial-time solvable for $H$-free graphs; otherwise it is \NP-hard or co-\NP-hard for $H$-free graphs.\\[-8pt]
\item [(iv)] Let $H\neq C_3\oplus P_1$. If $H\ssi P_4$ or $H\ssi \mbox{paw}$, then {\sc Contraction Blocker$(\omega$)} is polynomial-time solvable for $H$-free graphs, otherwise it is \NP-hard or co-\NP-hard for $H$-free graphs.\\[-8pt]
\item [(v)]
If $H\ssi P_1\oplus P_3$ or $H\ssi P_4$, then  {\sc Deletion Blocker$(\chi)$} is polynomial-time solvable for $H$-free graphs, otherwise it is \NP-hard or co-\NP-hard for  $H$-free graphs.\\[-8pt]
\item [(vi)]
If $H\ssi P_4$, then  {\sc Contraction Blocker$(\chi)$} is polynomial-time solvable for $H$-free graphs, otherwise it is \NP-hard 
or co-\NP-hard for  $H$-free graphs.
\end{itemize}
\end{theorem}
Statements (i), (ii), (iii), (v), (vi) of Theorem~\ref{t-mainmainmain} correspond to complete complexity dichotomies, whereas there is one missing case in statement~(iv). In particular we note that statements~(v) and~(vi) do not coincide for disconnected graphs~$H$.
We also observe from Theorem~\ref{t-mainmainmain}~(i) that {\sc Deletion Blocker($\alpha$)} is computationally hard for triangle-free graphs; in fact we will show co-\NP-hardness even if $d=k=1$. This in contrast to the problem being polynomial-time solvable for bipartite graphs, as shown in~\cite{BTT11,CWP11} (see also Table~\ref{t-thetable2}).

\subsection{Paper Organization} 

Section~\ref{sec:prelim} contains notation and terminology. 

Sections~\ref{sec:cob}--\ref{sec:chordal} contain the results mentioned in Part I.
To be more precise, Section~\ref{sec:cob} contains our results for cobipartite graphs, bipartite graphs and trees. 
In Sections~\ref{s-co} and~\ref{s-split}, we prove our results for cographs and split graphs, respectively. In Section~\ref{s-split} we also show that our \NP-hardness reduction for split graphs can be used to prove that the three contraction blockers problems, restricted to split graphs, are \W[1]-hard when parameterized by~$d$. The latter result means that for split graphs these problems are unlikely to be fixed-parameter tractable with parameter~$d$.
In Sections~\ref{sec:interval} and~\ref{sec:chordal} we prove our results for interval graphs and chordal graphs, respectively.

Sections \ref{sec:general} and~\ref{s-clas} contain the results mentioned in Part~II. In Section~\ref{sec:general} we first prove dichotomies for the three contraction blocker and three vertex blocker problems when we classify on basis of the size
of $\pi\in \{\alpha,\chi,\omega\}$. In the same section, we modify the hardness construction for {\sc $1$-Contraction Blocker($\omega$)}
to prove that {\sc $1$-Contraction Blocker($\omega$)} is \NP-complete even for 
$C_4$-free perfect graphs with $\omega=3$.
In Section~\ref{s-clas} we prove Theorem~\ref{t-mainmainmain}.

Section~\ref{s-con} contains a number of open problems and directions for future research.

\section{Preliminaries}\label{sec:prelim}
 
We only consider finite, undirected graphs that have no self-loops and no multiple edges; we recall that when we contract an edge no self-loops or multiple edges are created.
We refer to~\cite{Di05} or~\cite{West} for undefined terminology and to~\cite{DF99} for more on parameterized complexity.

Let $G=(V,E)$ be a graph. 
For a subset $S\subseteq V$, we let $G[S]$ denote the subgraph of $G$ {\it induced} by $S$, which has vertex set~$S$ and edge set $\{uv\in E\; |\; u,v\in S\}$. We write $H\ssi G$ if a graph $H$ is an induced subgraph of $G$.
Moreover, for a vertex $v\in V$, we write $G-v=G[V\setminus \{v\}]$ and for a subset $V'\subseteq V$ we write $G-V'=G[V\setminus V']$. 

For a set $\{H_1,\ldots,H_p\}$ of graphs, a graph $G$ is {\it $(H_1,\ldots,H_p)$-free} if $G$ has no induced subgraph isomorphic to a graph in $\{H_1,\ldots,H_p\}$; if $p=1$ we may write $H_1$-free instead of $(H_1)$-free.  
The {\it complement} of $G$  is the graph $\overline{G}=(V,\overline{E})$ with vertex set~$V$ and an edge between two vertices $u$ and $v$ if and only if~$uv\notin E$. 

Recall that the contraction of an edge $uv\in E$ removes the vertices $u$ and $v$ from a graph $G$ and replaces them by a new vertex that is made adjacent to precisely those vertices that were adjacent to $u$ or $v$ in $G$. This new graph will be denoted by $G\vert uv$. In that case we may also say that $u$ is \emph{contracted onto} $v$, and we use $v$ to denote the new vertex resulting from the edge contraction. The {\it subdivision} of an edge $uv\in E$ removes the edge $uv$ from $G$ and replaces it
by a new vertex $w$ and two edges $uw$ and $wv$.

Let $G$ and $H$ be two vertex-disjoint graphs.
The {\it join} operation $\otimes$ adds an edge between every vertex of $G$ and
every vertex of $H$. The {\it union} operation $\oplus$ takes the disjoint union of $G$ and $H$, that is, 
$G\oplus H=(V(G)\cup V(H),E(G)\cup E(H))$. We denote the disjoint union of $p$ copies of $G$ by $pG$.
For $n\geq 1$, the graph $P_n$ denotes the {\it path} on $n$ vertices, that is, $V({P_n})=\{u_1,\ldots,u_n\}$ and $E({P_n})=\{u_iu_{i+1}\; |\; 1\leq i\leq n-1\}$.
For $n\geq 3$, the graph $C_n$ denotes the {\it cycle} on $n$ vertices, that is,  $V({C_n})=\{u_1,\ldots,u_n\}$ and $E({C_n})=\{u_iu_{i+1}\; |\; 1\leq i\leq n-1\}\cup \{u_nu_1\}$. The graph $C_3$ is also called the {\it triangle}.
The {\it claw} $K_{1,3}$ is the 4-vertex star, that is, the graph with vertices $u$, $v_1$, $v_2$, $v_3$ and edges $uv_1$, $uv_2$, $uv_3$.

Let $G=(V,E)$ be a graph. A subset $K\subseteq V$ is called a {\it clique} of $G$ if any two vertices in $K$ are adjacent to each other.  The {\it clique number} $\omega(G)$ is the number of vertices in a maximum clique of $G$. 
 A subset $I\subseteq V$ is called an {\it independent set} of $G$ if any two vertices in $I$ are non-adjacent to each other.  The {\it independence number} $\alpha(G)$ is the number of vertices in a maximum independent set of $G$. For a positive integer $k$, a {\it $k$-colouring} of $G$ is a mapping $c: V\rightarrow\{1,2,\ldots,k\}$ such that $c(u)\neq c(v)$ whenever $uv\in E$.  The {\it chromatic number} $\chi(G)$ is the smallest integer $k$ for which $G$ has a $k$-colouring.  
A subset of edges $M\subseteq E$ is called a \textit{matching} if no two edges of $M$ share a common end-vertex. The {\it matching number} $\mu(G)$ is the number of edges in a maximum matching of a graph $G$. A vertex $v$ such that $M$ contains an edge incident with $v$ is \textit{saturated} by $M$; otherwise $v$ is \textit{unsaturated} by $M$.
A subset $S\subseteq V$ is a {\it vertex cover} of $G$ if every edge of $G$ is incident with at least one vertex of $S$.

The {\sc Coloring} problem is that of testing if a graph has a $k$-colouring for some given integer~$k$. The problems {\sc Clique} and {\sc Independent Set} are those of testing if a graph has a clique or independent set, respectively, of size at least~$k$. 
The {\sc Vertex Cover} problem is that of testing if a graph has a vertex cover of size at most~$k$. We need the following lemma at several places in our paper.

\begin{lemma}[\cite{Po74}]\label{l-po2}
{\sc Vertex Cover} is \NP-complete for $C_3$-free graphs.
\end{lemma}

An {\it interval graph} is a graph such that one can associate an interval of the real line with every vertex such that two vertices are adjacent if and only if the corresponding intervals intersect. A graph is {\it cobipartite} if it is the complement of a {\it bipartite} (2-colourable) graph. A graph is {\it chordal} if it contains no induced cycle on more than three vertices.
A graph is a {\it split graph} if it has a {\it split partition}, which is a partition of its vertex set into a clique~$K$ and an independent set~$I$.
Split graphs coincide with  $(2P_2,C_4,C_5)$-free graphs~\cite{FH77}.
A $P_4$-free graph is also called a {\it cograph}. 

A graph is {\it perfect} if the chromatic number of every induced subgraph equals the size of a largest clique in that subgraph. 
A  {\it hole} is an induced cycle on at least five vertices and an {\it antihole} is the complement of a hole. A hole or antihole is {\it odd} if it contains an odd number of vertices.
We need the following well-known theorem of Chudnovsky, Robertson, Seymour, and Thomas.
This theorem can also be used to verify that the other graph classes 
in Table~\ref{t-thetable2} are indeed subclasses of perfect graphs.

\begin{theorem}[Strong Perfect Graph Theorem \cite{CRST06}] 
\label{t-spgt}
A graph is perfect if and only if it contains no odd hole and no odd antihole.
\end{theorem}

\section{Cobipartite Graphs, Bipartite Graphs and Trees}\label{sec:cob}

We first consider the contraction blocker problems and then the deletion blocker problems.

\subsection{Contraction Blockers}

Our first result is a hardness result for cobipartite graphs that follows directly from a known result.

\begin{theorem}\label{t-firstco}
{\sc $1$-Contraction Blocker($\alpha$)} is \NP-complete for cobipartite graphs.
\end{theorem}

\begin{proof}
Golovach, Heggernes, van 't Hof and Paul~\cite{GHHP14} considered the {\sc $s$-Club Contraction} problem.
Recall that this problem takes as input a graph~$G$ and an integer $k$ and asks whether $G$ can be $k$-contracted into a graph with diameter at most~$s$ for some fixed integer~ $s$. They showed that {\sc $1$-Club Contraction} is \NP-complete even for cobipartite graphs. Graphs of diameter~1 are complete graphs, that is, graphs with independence number~1, whereas cobipartite graphs that are not complete have independence number~2. \qed
\end{proof}

We now focus on $\pi=\chi$ and $\pi=\omega$.
For our next result (Theorem~\ref{t-cohard}) we need some additional terminology.
A {\it biclique} is a complete bipartite graph,  which is {\it nontrivial} if it has at least one edge. 
A {\it biclique vertex-partition} of a graph $G=(V,E)$ is a set ${\mathcal S}$ of mutually vertex-disjoint bicliques in $G$ such that  every vertex of $G$ is contained in one  of the bicliques of ${\cal S}$. The {\sc Biclique Vertex-Partition} problem consists in testing whether a given graph $G$ has a biclique vertex-partition of size at most~$k$, for some positive integer $k$. Fleischner et al.~\cite{FMPS09} showed that this problem is \NP-complete even for bipartite graphs and $k=3$.

We are now ready to prove Theorem~\ref{t-cohard}.

\begin{theorem}\label{t-cohard}
For $\pi\in \{\chi,\omega\}$, {\sc Contraction Blocker($\pi$)} is \NP-complete for cobipartite graphs.
\end{theorem}

\begin{proof}
Since cobipartite graphs are perfect and closed under edge contractions, we may assume without loss of generality that $\pi=\chi$.
The problem is in \NP, as {\sc Coloring} is polynomial-time solvable on cobipartite graphs and then we can take the sequence of edge contractions as certificate.
We reduce from  {\sc Biclique Vertex-Partition}. Recall that this problem is \NP-complete even for bipartite graphs and $k=3$~\cite{FMPS09}.  As the problem is polynomial-time solvable for bipartite graphs and $k=2$ (see~\cite{FMPS09}), we may ask for a biclique vertex-partition of size exactly~3, in which each biclique is nontrivial.

Let ($G,3$) be an instance of {\sc Biclique Vertex-Partition}, where $G$ is a connected bipartite graph on $n$ vertices that has partition classes $X$ and $Y$. We claim that $G$ has a biclique vertex-partition consisting of three non-trivial bicliques if and only if $\overline{G}$ can be $(n-6)$-contracted into a graph $G'$ with $\chi(G')\leq 3$ (so $d=\chi(\overline G)-3$).

First suppose that $G$ has a biclique vertex-partition ${\cal S}$ of size 3. Let $S_1,S_2,S_3$ be the three (nontrivial) bicliques in ${\cal S}$. Let $A_i,B_i$ be the two bipartition classes of $S_i$ for $i=1,2,3$. So, in $\overline{G}$, we have that $A_1,A_2,A_3$, $B_1,B_2,B_3$ are six cliques that partition the vertices of $\overline{G}$, and moreover, there is no edge between a vertex of $A_i$ and  a vertex of $B_i$, for $i=1,2,3$. In $\overline{G}$ we contract each clique $A_i$ to a single vertex that we give colour~$i$, and we contract each clique $B_i$ to a single vertex that we give colour~$i$ as well. In this way we have obtained a 6-vertex graph $G'$ (so the number of contractions is $n-6$) with a 3-colouring. Thus, $\chi(G')\leq 3$.

Now suppose that $\overline{G}$ can be $(n-6)$-contracted into a graph $G'$ with $\chi(G')\leq 3$.  We first observe that the class of cobipartite graphs is closed under taking edge contractions; indeed, if $e$ is an edge connecting two vertices of the same partition class, then contracting $e$ results in a smaller clique, and if $e$ is an edge connecting two vertices of two different partition classes, then contracting $e$ is equivalent to removing one of its end-vertices and making its other end-vertex adjacent to every other vertex in the resulting graph.
 
As the class of cobipartite graphs is closed under taking contractions, $G'$ is cobipartite. As cobipartite graphs have independence number at most~2, each colour class in a colouring of $G'$ must have size at most~2. Consequently, $G'$ must have exactly six vertices $a_1,a_2,a_3$, $b_1,b_2,b_3$ such that $a_1,a_2,a_3$ form a clique, $b_1,b_2,b_3$ form a clique, and moreover,  $a_i$ and $b_i$ are not adjacent, for $i=1,2,3$. This means that we did not contract an edge $uv$ with $u\in X$ and $v\in Y$ (as the resulting vertex would be adjacent to all other vertices).  Hence, we may assume without loss of generality that for $i=1,2,3$, each $a_i$ corresponds to a set of vertices $A_i\subset X$ (that we contracted into the single vertex $a_i$) and that  each $b_i$ corresponds to a set of vertices $B_i\subset Y$ (that we contracted into the single vertex $b_i$). As each pair $a_i,b_i$ is non-adjacent,  no vertex of $A_i$ is adjacent to a vertex of $B_i$. Consequently, in $G$, we find that each set $A_i\cup B_i$ induces a biclique. Hence, the sets $A_1\cup B_1$, $A_2\cup B_2$ and $A_3\cup B_3$ form a biclique vertex-partition of $G$ that has size~3.\qed
\end{proof}

We now assume that $d$ is fixed. We show that  {\sc $d$-Contraction Blocker($\pi$)} becomes polynomial-time solvable on cobipartite graphs for $\pi\in \{\chi,\omega\}$. For $\pi=\chi$, we can prove this even for the class of graphs with independence number at most~2, or equivalently, the class of $3P_1$-free graphs, which properly contains the class of cobipartite graphs.

\begin{theorem}\label{t-3p1}
For any fixed $d\geq 0$, the {\sc $d$-Contraction Blocker($\chi$)} problem can be solved in polynomial time for $3P_1$-free graphs.
\end{theorem}

\begin{proof}
Let $G$ be a graph with $\alpha(G)\leq 2$. Consider a colouring with $\chi(G)$ colours. The size of every colour class is at most $2$. Hence every subgraph of $G$ induced by two colour classes has at most $4$ vertices, and as such has a spanning forest with in total at most $3$ edges.
This means that we can contract two colour classes to an independent set (that is, to a new colour class) by using at most $3$ contractions. This observation gives us the following algorithm.
We guess a set of at most $3$ contractions. Afterward we decrease $d$ by~1 and repeat this procedure until $d=0$. For each resulting graph $G'$ we check whether $\chi(G')\leq \chi(G)-d$. If so, then the algorithm returns a yes-answer and otherwise a no-answer.

Let $m$ be the number of edges of $G$. Then the total number of guesses is at most $m^{3d}$, which is polynomial as $d$ is fixed. Because {\sc Coloring} is polynomial-time solvable on graphs with independence number at most~2  and this class is closed under edge contractions, our algorithm runs in polynomial time. \qed
\end{proof}

\begin{corollary}
Let $\pi\in \{\chi,\omega\}$. For any fixed $d\geq 0$, the {\sc $d$-Contraction Blocker($\pi$)} problem can be solved in polynomial time on cobipartite graphs.
\end{corollary}

\begin{proof}
For $\pi=\chi$ this follows immediately from 
Theorem~\ref{t-3p1}.  As cobipartite graphs are perfect and closed under edge contraction, we obtain the same result for
$\pi=\omega$.\qed
\end{proof}

We now consider the class of bipartite graphs.
If $\pi\in \{\chi,\omega\}$, then {\sc Contraction Blocker($\pi$)} is trivial for bipartite graphs (and thus also for trees).
To the contrary, for $\pi=\alpha$, we will show 
that {\sc Contraction Blocker($\pi$)} is \NP-hard for bipartite graphs.
The complexity of {\sc $d$-Contraction Blocker($\alpha$)} remains open for bipartite graphs.
Bipartite graphs are not closed under edge contraction. Therefore membership to \NP\ cannot be established by taking a sequence of edge contractions as the certificate, even though due to K\"onig's Theorem (see, for example, \cite{Di05}), {\sc Independent Set} is
 polynomial-time solvable for bipartite graphs.

\begin{theorem}
\label{thm:bipartite}
{\sc Contraction Blocker($\alpha$)} is \NP-hard on bipartite graphs. 
\end{theorem}

\begin{proof}
We know from Theorem \ref{t-firstco} that {\sc $1$-Contraction Blocker($\alpha$)} is \NP-complete on cobipartite graphs. Consider a cobipartite graph $G$ with $m$ edges and an integer $k$, which together form an instance of {\sc $1$-Contraction Blocker($\alpha$)}. Subdivide each of the $m$ edges of $G$ in order to obtain a bipartite graph $G'$. We claim that $(G,k)$ is a yes-instance of {\sc 1-Contraction Blocker($\alpha$)} if and only if $(G',\alpha(G')-1,k+m)$ is a yes-instance of {\sc Contraction Blocker($\alpha$)}. 

First suppose that $(G,k)$ is a yes-instance of {\sc 1-Contraction Blocker($\alpha$)}. In $G'$ we first perform $m$ edge contractions to get $G$ back.
We then perform $k$ edge contractions to get independence number $1=\alpha(G')-(\alpha(G')-1)$. Hence, $(G',\alpha(G')-1,k+m)$ is a yes-instance of 
{\sc Contraction Blocker($\alpha$)}. 

Now suppose that $(G',\alpha(G')-1,k+m)$ is a yes-instance of {\sc Contraction Blocker($\alpha$)}. Then there exists a sequence of $k+m$ edge contractions that transform $G'$ into a complete graph $K$. We may assume that 
$K$ has size at least~4 (as we could have added without loss of generality three dominating vertices to $G$ without increasing $k$). As $K$ has size at least~4, each subdivided edge must be contracted back to the original edge again. This operation costs $m$ edge contractions, so we contract $G$ to $K$ using at most $k$ edge operations. Hence, $(G,k)$ is a yes-instance of {\sc 1-Contraction Blocker($\alpha$)}. This proves the claim and hence the theorem.\qed
\end{proof}

We complement Theorem~\ref{thm:bipartite} by showing that {\sc Contraction Blocker($\alpha$)} is linear-time solvable on trees. In order to prove this result
we make a connection to the matching number~$\mu$ of a graph. 

\begin{theorem}
\label{thm:trees}
{\sc Contraction Blocker($\alpha$)} is linear-time solvable on trees. 
\end{theorem}

\begin{proof}
Let $(T,d,k)$ be an instance of {\sc Contraction Blocker($\alpha$)}, where $T$ is a tree on $n$ vertices. We first describe our algorithm and prove its correctness. Afterwards, we analyze its running time. Throughout the proof let $M$ denote a maximum matching of $T$.

As $\alpha(T)+\mu(T)=n$ by K\"{o}nig's Theorem (see, for example, \cite{Di05}), we find that $(T,d,k)$ is a no-instance if $d>n-\mu(T)$.  Assume that $d\le n-\mu(T)$.
We observe that trees are closed under edge contraction. Hence, contracting an edge of $T$ results in a new tree $T'$. Moreover, $T'$
has $n-1$ vertices and the edge contraction neither increased the independence number nor the matching number.
As $\alpha(T)+\mu(T)=n$ and similarly  $\alpha(T')+\mu(T')=n-1$, this means that either $\alpha(T')=\alpha(T)-1$ or $\mu(T')=\mu(T)-1$. 

First suppose that $d\le n-2\mu(T)$.
There are exactly $\sigma(T)=n-2\mu(T)$ vertices that are unsaturated by $M$. Let $uv$ be an edge, such that $u$ is unsaturated. 
As $M$ is maximum, $v$ must be saturated. Then, by contracting $uv$, we obtain a tree $T'$ such that $\mu(T')=\mu(T)$. 
It follows from the above that $\alpha(T')=\alpha(T)-1$. Say that we contracted $u$ onto $v$. 
Then in $T'$ we have that $v$ is saturated by $M$, which is a maximum matching of $T'$ as well.  Thus, if $d\le n-2\mu(T)$, contracting $d$ edges, one of the end-vertices of which is unsaturated by $M$, yields a tree $T'$ with $\mu(T')=\mu(T)$ and $\alpha(T')=\alpha(T)-d$. Since an edge contraction reduces the independence number by at most~1, it follows that this is optimal. Hence, as $d\le n-2\mu(T)$, we find that $(G,T,k)$ is a yes-instance if $k\ge d$ and a no-instance if $k<d$. 

Now suppose that $d>n-2\mu(T)$. Suppose that we first contract the $n-2\mu(T)$ edges that have exactly one end-vertex that is unsaturated by $M$. It follows from the above that this yields a tree $T'$ with $\mu(T')=\mu(T)$ and $\alpha(T')=\alpha(T)-(n-2\mu(T))$. Since $T'$ does not contain any unsaturated vertex, $M$ is a perfect matching of $T'$. Then, contracting any edge in $T'$ results in a tree $T''$ with $\mu(T'')=\mu(T')-1$ and thus, $\alpha(T'')=\alpha(T')$. If we contract an edge $uv\in M$, the resulting vertex $uv$ is unsaturated by  $M'=M\setminus\{uv\}$ in $T''$. Hence, as explained above, if in addition we contract now an edge $(uv)w$, we obtain a tree $T'''$ with $\alpha(T''')=\alpha(T'')-1$ and $\mu(T''')=\mu(T'')$. Repeating this procedure, we may reduce the independence number of $T$ by $d$ with $n-2\mu(T)+2(d-n+2\mu(T))=2(d+\mu(T))-n$ edge contractions. Below we show that this is optimal.

Suppose that we contract $p$ edges in  $T$. Let $T'$ be the resulting tree. We have $\alpha(T')+\mu(T')=n-p$. As $\mu(T')\le {1\over 2}(n-p)$, this means that $\alpha(T')\ge {1\over 2}(n-p)$. If $p<2(d+\mu(T))-n$ we have $-{p\over 2}>-(d+\mu(T))+{n\over 2}$,  and thus 
\[\begin{array}{lcl}
\alpha(T') &\ge &{1\over 2}(n-p)\\[3pt]
&>&{n\over 2}-d-\mu(T)+{n\over 2}\\[3pt]
&=&\alpha(T)-d.
\end{array}\]
So at least $2(d+\mu(T))-n$ edge contractions are necessary to decrease the independence number by $d$.  It remains to check if $k$ is sufficiently high for us to allow this number of edge contractions. 

As we can find a maximum matching of tree $T$ (and thus compute $\mu(T)$) in $O(n)$ time by using the algorithm of Savage~\cite{Sa80}, our algorithm runs in $O(n)$ time.
\qed
\end{proof}

\noindent
{\bf Remark 1.} By K\"onig's Theorem, we have that $\alpha(G)+\mu(G)=|V(G)|$  for any bipartite graph $G$, but we can only use the proof of Theorem~\ref{thm:trees}
to obtain a result for trees  for the following reason: trees form the largest subclass of (connected) bipartite graphs that are closed under edge contraction, and this property plays a crucial role in our proof.

\subsection{Deletion Blockers}

We first show that all three deletion blocker problems are polynomial-time solvable for bipartite graphs (and thus for trees).
It is known already that {\sc Deletion Blocker($\alpha$)} is polynomial-time solvable for bipartite graphs~\cite{BTT11,CWP11}. Hence it suffices to prove that the same holds for {\sc Deletion Blocker($\pi$)} when $\pi\in\{\chi, \omega\}$.
In order to do so we need the following relation between {\sc $1$-Deletion Blocker($\omega$)} and {\sc Vertex Cover}.

\begin{proposition}
\label{p-vc}
Let $G$ be a graph with at least one edge and let $k \geq 1$ be an integer. Then $(G,\omega(G)-1,k)$ is a yes-instance of {\sc Deletion Blocker($\omega$)} if and only if $(G,k)$ is a yes-instance of {\sc Vertex Cover}.
\end{proposition}

\begin{proof}
Let $G=(V,E)$ be a graph with $|E|\geq 1$. Thus, $\omega(G) \geq 2$. Let $k \geq 1$ be an integer. First suppose that $(G,k)$ is a yes-instance of {\sc Vertex Cover}, that is, $G$ has a vertex cover $V'$ of size at most $k$. So, every edge of $G$ is incident to at least one vertex of $V'$.  Then, deleting all vertices of $V'$ yields a graph $G'$ with no edges. This means that $\omega(G') \leq 1$, and thus $(G,\omega(G)-1,k)$ is a yes-instance for {\sc Deletion Blocker($\omega$)}.
Now suppose that $(G,\omega(G)-1,k)$ is a yes-instance of {\sc Deletion Blocker($\omega$)}. Then there exists a set $V'\subseteq V$ of size $|V'|\leq k$ such that $\omega(G-V') \leq 1$.  This implies that $G-V'$ has no edges. Thus $V'$ is a vertex cover of $G$ of size at most $k$. So, $(G,k)$ is a yes-instance for {\sc Vertex Cover}. \qed
\end{proof}

Proposition~\ref{p-vc} has the following corollary, which we will apply in this section and at some other places in our paper.

\begin{corollary}
\label{c-vc}
Let $G$ be a triangle-free graph with at least one edge and let $k \geq 1$ be an integer. Then $(G,k)$ is a yes-instance of {\sc 1-Deletion Blocker($\omega$)} if and only if $(G,k)$ is a yes-instance of {\sc Vertex Cover}.
\end{corollary}

We are now ready to prove the following result.

\begin{theorem}\label{thm:bip}
For $\pi\in \{\chi,\omega\}$, {\sc Deletion Blocker($\pi$)} can be solved in polynomial time on bipartite graphs.
\end{theorem}

\begin{proof}
As bipartite graphs are perfect and closed under vertex deletion, the problems {\sc Deletion Blocker($\omega$)} and {\sc Deletion Blocker($\chi$)} are equivalent.
Therefore, we only have to consider the case where $\pi=\omega$. As bipartite graphs have clique number at most~2, {\sc Deletion Blocker($\omega$)} and {\sc $1$-Deletion Blocker($\omega$)} are equivalent. 
As bipartite graphs are triangle-free, we can apply Corollary~\ref{c-vc}.
To solve \textsc{Vertex Cover} on bipartite graphs, K\"{o}nig's Theorem  tells us that it suffices to find a maximum matching, which takes $O(n^{2.5})$ time on $n$-vertex bipartite graphs~\cite{HK73}.\qed
\end{proof}

We now consider the 
the class of cobipartite graphs.
It is known that {\sc Deletion Blocker($\pi$)} is polynomial-time solvable on cobipartite graphs if $\pi\in\{\omega, \chi\}$~\cite{CWP11}.
Hence we only have to deal with the case $\pi=\alpha$. For this case we prove the following result, which follows immediately from Theorem~\ref{thm:bip}.

\begin{theorem}
{\sc Deletion Blocker($\alpha$)} can be solved in polynomial time on cobipartite graphs.
\end{theorem}

\section{Cographs}\label{s-co}

It is well known (see for example~\cite{BLS99}) that  a graph $G$ is a cograph if and only if $G$ can be generated from $K_1$ by a sequence of operations, where each operation is either a join or a union operation. 
Recall from Section~\ref{sec:prelim} that we denote these operations by $\otimes$ and $\oplus$, respectively.
Such a sequence corresponds to  a decomposition tree $T$, which has the following properties:
\begin{itemize}
\item [1.]  its root $r$ corresponds to the graph $G_r=G$;
\item [2.]   every leaf $x$ of $T$ corresponds to exactly one vertex of $G$, and vice versa, implying that $x$ corresponds to a unique single-vertex graph $G_x$;
\item [3.]  every internal node $x$ of $T$ has at least two children, is either labeled $\oplus$ or $\otimes$, and corresponds to an induced subgraph $G_x$ of $G$ defined as follows:
\begin{itemize}
\item if $x$ is a $\oplus$-node, then $G_x$ is the disjoint union of all graphs $G_y$ where $y$ is a child of $x$;
\item if $x$ is a $\otimes$-node, then $G_x$ is the join of all graphs $G_y$ where $y$ is a child of $x$.
\end{itemize}
\end{itemize}
A cograph $G$ may have more than one such tree but has exactly one unique tree~\cite{CHS81}, called the {\it cotree} $T_G$ of $G$, if the following additional property is required: 
\begin{itemize} 
\item [4.] Labels of internal nodes on the (unique) path from any leaf to $r$ alternate between $\oplus$ and $\otimes$.
\end{itemize}

Note that $T_G$ has $O(n)$ vertices. For our purposes we must modify $T_G$ by applying the following known procedure
(see for example~\cite{BM93}). Whenever an internal node~$x$ of $T_G$ has more than two children $y_1$ and $y_2$, we remove the edges $xy_1$ and $xy_2$ and add a new vertex $x'$ with 
edges $xx'$, $x'y_1$ and $x'y_2$. If~$x$ is a $\oplus$-node, then $x'$ is a $\oplus$-node, and if $x$ is a $\otimes$-node, then $x'$ is a 
$\otimes$-node. Applying this rule exhaustively yields a tree in which each internal node has exactly two children. We denote this tree
by $T_G'$. Because $T_G$ has $O(n)$ vertices, modifying $T_G$ into $T_G'$ takes linear time.

Corneil, Perl and Stewart~\cite{CPS85} proved that the problem of deciding whether a graph with $n$ vertices and $m$ edges is a cograph can be solved in time $O(n+m)$. They also showed that in the same time it is possible to construct its cotree (if it exists).
As modifying $T_G$ into $T_G'$ takes $O(n+m)$ time, we obtain the following lemma.

\begin{lemma}\label{l-cotree}
Let $G$ be a graph with $n$ vertices and $m$ edges.  Deciding if  $G$ is a cograph and constructing $T_G'$ (if it exists) can be done in time $O(n+m)$.
\end{lemma}

For two integers $k$ and $\ell$ we say that a graph $G$ can be {\it $(k,\ell)$-contracted} into a graph $H$ if $G$ can be modified into $H$ by a sequence containing
$k$ edge contractions and $\ell$ vertex deletions. Note that cographs are closed under edge contraction and under vertex deletion. 
In fact, to prove our results for cographs, we will prove the following more general result.

\begin{theorem}\label{t-coalpha}
Let $\pi\in \{\alpha,\chi,\omega\}$. The problem of determining the largest integer~$d$ such that a cograph $G$ with $n$ vertices and $m$ edges can be $(k,\ell)$-contracted into a cograph $H$ with $\pi(H)\leq \pi(G)-d$  can be solved in  $O(n^2+mn+(k+\ell)^3n)$ time. 
\end{theorem}

\begin{proof}
First consider $\pi=\alpha$.
Let $G$ be a cograph with $n$ vertices and $m$ edges and let $k,\ell$ be two positive integers. We first construct $T_G'$. We then consider each node of $T_G'$ by following a bottom-up approach starting at the leaves of $T_G'$ and ending in its root $r$. 

Let $x$ be a node of $T_G'$. Recall that $G_x$ is the subgraph of $G$
induced by all vertices that corresponds to leaves in the subtree of $T_G'$ rooted at $x$. 
With node~$x$ we associate a table  that records the following data: for each pair of integers $i,j\geq 0$ with $i+j\leq k+\ell$  
 we compute the largest integer $d$ 
such that $G_x$ can be $(i,j)$-contracted into a graph $H_x$ with $\alpha(H_x)\leq \alpha(G_x)-d$. We denote this integer $d$ by 
$d(i,j,x)$. Let  $i,j\geq 0$ with $i+j\leq k+\ell$.

\medskip
\noindent
{\bf Case 1.} $x$ is a leaf.\\
Then $G_x$ is a 1-vertex graph meaning that $d(i,j,x)=0$ if $j=0$, whereas $d(i,j,x)=1$ if $j\geq 1$.

\medskip
\noindent
{\bf Case 2.} $x$ is a $\oplus$-node.\\
Let $y$ and $z$ be the two children of $x$. Then, as $G_x$ is the disjoint union of $G_y$ and $G_z$, we find that
$\alpha(G_x)=\alpha(G_{y})+\alpha(G_{z})$. 
Hence, we have 
\[
\begin{array}{lcrl}
d(i,j,x) &= &\max&\{\alpha(G_x)-(\alpha(G_{y})-d(a,b,y)+\alpha(G_{z})-d(i-a,j-b,z))\; |\;\\ &&& \;\;0\leq a\leq i, 0\leq b\leq j\}\\[4pt]
&= &\max&\{d(a,b,y)+d(i-a,j-b,z)\; |\; 0\leq a\leq i, 0\leq b\leq j\}.
\end{array}
\]
\noindent
{\bf Case 3.} $x$ is a $\otimes$-node.\\ 
Since $x$ is a $\otimes$-node, $G_x$ is connected and as such has a spanning tree $T$.
If $i+j\geq |V(G_x)|$ and $j\geq 1$, then we can contract $i$ edges of $T$ in the graph $G_x$ followed by $j$ vertex deletions. As each operation will reduce $G_x$ by exactly one vertex, this results in the empty graph. Hence, $d(i,j,x)=\alpha(G_x)$.
From now on assume that $i+j<|V(G_x)|$ or $j=0$. As such, any graph we can obtain from $G_x$ by using $i$ edge contractions and $j$ vertex deletions is non-empty and hence has
independence number at least~1.

Let $y$ and $z$ be the two children of $x$. Then, as $G_x$ is the join of $G_{y}$ and $G_{z}$, 
we find that $\alpha(G_x)=\max\{\alpha(G_y),\alpha(G_z)\}$. 
In order to determine $d(i,j,x)$ we must do some further analysis.
Let $S$ be a sequence that consists of $i$ edge contractions and $j$ vertex deletions of $G_x$ such that applying $S$ on $G_x$ results in a graph $H_x$
with $\alpha(H_x)=\alpha(G_x)-d(i,j,x)$. 
We partition $S$ into five sets $S_y^e$, $S_z^e$, $S_{yz}^e$, $S_y^v$, $S_z^v$,  respectively, as follows.
Let $S_y^e$ and $S_z^e$ be the set of contractions of edges with both end-vertices in $G_y$ and with both end-vertices in $G_z$, respectively.
Let $S_{yz}^e$ be the set of contractions of edges with one end-vertex in $G_y$ and the other one in $G_z$.
Let $a_y=|S_y^e|$ and let $a_z=|S_z^e|$. Then $|S_{yz}^e|=i-a_y-a_z$.
Let $S_y^v$ and $S_z^v$ be the set of deletions of vertices in $G_y$ and $G_z$, respectively. 
Let $b=|S_y^v|$. Then $|S_z^v|=j-b$.
We distinguish between two cases. 

\medskip
\noindent
First assume that $S^e_{yz}=\emptyset$.
Then $a_y+a_z=i$.
Let $H_y$ be the graph obtained from $G_y$ after applying the subsequence of $S$, consisting of operations in $S_y^e\cup S_y^v$, on $G_y$.
Let $H_z$ be defined analogously.
Then we have
\[
\begin{array}{lcl}
\alpha(H_x) & = &\max\{\alpha(H_y),\alpha(H_z)\}\\[4pt]
&= &\max\{\alpha(G_y)-d(a_y,b,y),\alpha(G_z)-d(a_z,j-b,z)\}\\[4pt]
&= &\max\{\alpha(G_y)-d(a_y,b,y),\alpha(G_z)-d(i-a_y,j-b,z)\},
\end{array}
\]
where the second equality follows from the definition of $S$.

\medskip
\noindent
Now assume that $S^e_{yz}\neq \emptyset$. 
Recall that  $i+j<|V(G_x)|$ or $j=0$.
Hence $\alpha(H_x)\geq 1$.
Our approach is based on the following observations.

First, contracting an edge with one end-vertex in $G_y$ and the other one in $G_z$ is equivalent to removing these two end-vertices and introducing a new 
vertex that is adjacent to all other vertices of $G_x$ (such a vertex is said to be {\it universal}).

Second, assume that $G_y$ contains two distinct vertices $u$ and $u'$ and that $G_z$ contains two distinct vertices $v$ and $v'$. Now suppose that we are to
contract two edges from $\{uv,uv',u'v,u'v'\}$. Contracting two edges of this set that have a common end-vertex, say edges $uv$ and $uv'$, is equivalent to deleting $u,v,v'$ from $G_x$ and introducing a new universal vertex. Contracting two edges with no common end-vertex, say $uv$ and $u'v'$,
is equivalent to deleting all four vertices $u,u',v,v'$ from $G_x$ and introducing two new universal vertices.
Because the two new universal vertices in the latter choice are adjacent, whereas the vertex $u'$ may not be universal after making the former choice, the latter choice decreases the independence number by the same or a larger value than the former choice.
Hence, we may assume without loss of generality that the latter choice happened. More generally, the contracted edges 
with one end-vertex in $G_y$ and the other one in $G_z$ can be assumed to form a matching. 
We also note that introducing a new universal vertex to a graph does not introduce any new independent set other than the singleton
set containing the vertex itself.

We conclude that each edge contraction in $S^e_{yz}$ may be considered to be equivalent to deleting one vertex from $G_y$ and one from $G_z$ and introducing a new 
universal vertex. If one of the two graphs $G_y$ or $G_z$ becomes empty in this way, then an edge contraction in $S^e_{yz}$ can  be considered to
be equivalent to the deletion of a vertex of the other one. Finally, if both sets $G_y$ and $G_z$ become empty, then we can stop as in that case $H_x$ has independence number~1 (which we assumed was the smallest value of $\alpha(H_x)$).

By the above observations and the definition of $S$ we find that
$$\alpha(H_x) =\max\{1,\alpha(G_y)-d(a_y,b+i-a_y-a_z,y),\alpha(G_z)-d(a_z,j-b+i-a_y-a_z,z)\}.$$

Hence we can do as follows. We consider all tuples $(a_y,b)$ with $0\leq a_y\leq i$ and $0\leq b\leq j$ and compute 
$\max\{\alpha(G_y)-d(a_y,b,y),\alpha(G_z)-d(i-a_y,j-b,z)\}$. Let $\alpha_x'$ be the minimum value over all values found.
We then consider all tuples $(a_y,a_z,b)$ with $a_y\geq 0$, $a_z\geq 0$, $a_y+a_z\leq i$ and $0\leq b\leq j$ and compute 
$\max\{1,\alpha(G_y)-d(a_y,b+i-a_y-a_z,y),\alpha(G_z)-d(a_z,j-b+i-a_y-a_z,z)\}$. Let $\alpha_x''$ be the minimum value over all values found.
Then $d(i,j,x)=\alpha(G_x)-\min\{\alpha_x',\alpha_x''\}$.

\medskip
\noindent
After reaching the root $r$, we let our algorithm return the integer $d(k,\ell,r)$. By construction, $d(k,\ell,r)$ is the largest integer such that $G=G_r$ can be $(k,\ell)$-contracted into a graph $H$ with $\alpha(H)\leq \alpha(G)-d(k,\ell,r)$.
We are left to analyze the running time.

Constructing $T_G'$ can be done in $O(n+m)$ time by Lemma~\ref{l-cotree}.
We now determine the time it takes to compute one entry $d(i,j,x)$ in the table associated with a node $x$.
It takes linear time to compute the independence number of a cograph\footnote{For a cograph $G$,  compute $T_G'$ and use the formula
$\alpha(G_x)=\alpha(G_{y})+\alpha(G_{z})$ if $x$ is a $\oplus$-node with children $y$ and $z$ and  $\alpha(G_x)=\max\{\alpha(G_y),\alpha(G_z)\}$ otherwise.
Alternatively, see for example~\cite{CHMW87} for a linear-time algorithm on a superclass of cographs.}.
The total number of tuples $(a_y,b)$ and $(a_y,a_z,b)$ that we need to consider is $O((k+\ell)^3)$.
Note that the table associated with a node $x$ has $O((k+\ell)^2)$ entries but that we only have to compute $\alpha(G_x)$ once.
Hence, it takes $O(n+m+(k+\ell)^3)$ time to construct a table for a node.
As~$T_{G'}$ has~$O(n)$ vertices, the total running time is $O(n+m)+O(n(n+m+(k+\ell)^3))=O(n^2+mn+(k+\ell)^3n)$.

\medskip
\noindent
Now consider $\pi=\chi$.
Note that we cannot consider the complement of a cograph (which is a cograph) because an edge contraction in a graph does not correspond to an edge contraction in its complement. However, we can re-use the previous proof after making a few modifications.
Let $G$ be a cograph with $n$ vertices and $m$ edges and let $k,\ell$ be two positive integers.
We follow the same approach as in the proof for $\pi=\alpha$. We only have to swap Cases~2 and~3 after
observing that $\chi(G_x)=\max\{\chi(G_{y}),\chi(G_{z})\}$ if $x$ is a $\oplus$-node with $y$ and $z$ as its two children
and $\chi(G_x)=\chi(G_{y})+\chi(G_{z})$ if $x$ is a $\otimes$-node. We can use the same arguments as used in the proof
for $\pi=\alpha$ for the running time analysis as well;
we only have to observe that it takes $O(n+m)$ time to compute the chromatic number of a cograph (using the same arguments as before or by using another algorithm of~\cite{CHMW87}).

\medskip
\noindent
Finally consider $\pi=\omega$.
As cographs are perfect and closed under edge contractions, the proof follows immediately from the corresponding result for $\pi=\chi$.\qed
\end{proof}

\begin{corollary}\label{c-coalpha}
For $\pi\in \{\alpha,\chi,\omega\}$, both the {\sc Contraction Blocker($\pi$)} problem and the {\sc Deletion Blocker($\pi$)} problem can be solved in polynomial time for cographs.
\end{corollary}

\begin{proof}
We use Theorem~\ref{t-coalpha} after setting $\ell=0$ for {\sc Contraction Blocker($\pi$)} and $k=0$ for {\sc Deletion Blocker($\pi$)}.\qed
\end{proof}

\section{Split Graphs}\label{s-split}

A split partition $(K,I)$ of a split graph is {\it minimal} if $I\cup \{v\}$ is not an independent set for all $v\in K$, in other words every vertex $v\in K$ is adjacent to some vertex $u\in I$. Note that for a minimal split partition $(K,I)$ we have $\alpha(G)=\vert I\vert$. A split partition $(K,I)$ is {\it maximal} if $K\cup \{v\}$ is not a clique for all $v\in I$,
 in other words every vertex $v\in I$ is non adjacent to at least one vertex $u\in K$. Note that for a maximal split partition $(K,I)$ we have $\omega(G)=\chi(G)=\vert K\vert$.
We first show the following result.

\begin{theorem}\label{t-splitalpha}
Let $\pi \in\{\alpha,\chi,\omega\}$.
For any 
fixed $d\geq 0$,  the $d$-{\sc Contraction Blocker$(\pi)$} problem is polynomial-time solvable on split graphs.
\end{theorem}

\begin{proof}
First consider $\pi=\alpha$.
Let $(G,k)$ be an instance of {\sc $d$-Contraction Blocker$(\alpha)$} where $G=(V,E)$ is a split graph. Let $(K,I)$ be a minimal split partition of $G$.
Let $I'$ be the set of vertices in $I$ that have at least one neighbour in $K$, and let $I''=I\setminus I'$. Because $G$ is a split graph, all vertices of $I'$ belong to the same connected component $D$
of $G$. Moreover, we have $\alpha(G)=|I|=|I'|+|I''|=\alpha(D)+|I''|$. 

First suppose that $|I'|\leq d$. 
For $(G,k)$ to be a yes-instance, $G$ must be contracted into a graph $G'$ with $\alpha(G')\leq \alpha(G)-d=|I'|+|I''|-d\leq |I''|$. This means that we must contract $D$ into the empty graph, which is not possible. Hence, $(G,k)$ is a no-instance in this case. Hence, we may assume without loss of generality that $|I'|\geq d+1$.

Suppose that $k\geq d+1$. If $k\geq |I'|$, then we contract every vertex of $I'$ onto a neighbour in $K$. In this way we have $k$-contracted $G$ into a graph $G'$ with $\alpha(G')=|I''|+1=
|I'|+|I''|-(|I'|-1)\leq |I'|+|I''|-d=\alpha(G)-d$. So, $(G,k)$ is a yes-instance in this case.
If $k\leq |I'|-1$, we contract each vertex of an arbitrary subset of $k$ vertices of $I'$ onto a neighbour in $K$. 
In this way we have $k$-contracted $G$ into a graph $G'$ with $\alpha(G')\leq |I'|-k+1+|I''|\leq |I'|+|I''|-d=\alpha(G)-d$. So, $(G,k)$ is a yes-instance in this case as well.

If $k\leq d$, then we consider all possible sequences of at most $k$ edge contractions. This takes time $O(|E(G)|^k)$, which is polynomial as $d$, and consequently $k$, is fixed.
For every such sequence we check in polynomial time whether the resulting graph has stability number at most $\alpha(G)-d$. As split graphs are closed under edge contraction and moreover are chordal graphs, the latter can be verified in linear time (see~\cite{Go80}).

\medskip
\noindent
Now let $\pi=\chi$.
Let $(G,k)$ be an instance of {\sc $d$-Contraction Blocker$(\chi)$} where $G=(V,E)$ is a split graph.

\medskip
\noindent
{\bf Case 1.} $\chi(G)\leq d$.\\
For $(G,k)$ to be a yes-instance, $G$ must be $k$-contracted into a graph $G'$ with $\chi(G')\leq \chi(G)-d\leq 0$. The only graph with chromatic number at most~$0$, is the empty graph. However, a non-empty graph cannot be contracted to an empty graph. Hence, $(G,k)$ is a no-instance in this case.

\medskip
\noindent
{\bf Case 2.} $\chi(G)=d+1$.\\ 
For $(G,k)$ to be a yes-instance, $G$ must be $k$-contracted into a graph $G'$ with $\chi(G')\leq \chi(G)-d=1$. Hence, every connected component of $G'$ must consist of exactly one vertex. If $G$ has no connected components with edges, then $(G,k)$ is a yes-instance. Otherwise,
because $G$ is a split graph, $G$ has exactly one connected component $D$ containing one or more edges. In that case, $(G,k)$ is a yes-instance if and only if $k\geq |V(D)|-1$; this can be checked in constant time.

\medskip
\noindent
{\bf Case 3.} $\chi(G)\geq d+2$.\\
First, assume that $k<d$.
Because every edge contraction reduces the chromatic number by at most~1, 
$(G,k)$ is a no-instance.

Second, assume that $k=d$. 
We consider all possible sequences of at most $k$ edge contractions. This takes time $O(|E(G)|^k)$, which is polynomial as $d$, and consequently $k$, is fixed.
For every such sequence we check in polynomial time whether the resulting graph has chromatic number at most $\chi(G)-d$. As split graphs are closed under edge contractions and moreover are chordal graphs, the latter can be verified in polynomial time  (see \cite{Go80}).

Third, assume that $k>d$. We claim that $(G,k)$ is a yes-instance. This can be seen as follows.
Let $(K,I)$ be a maximal split partition of $G$. 

If $k<|K|$, then we contract $k$ arbitrary edges of $K$. The resulting graph $G'$ has a split partition $(K',I)$ with $|K'|= |K|-k\leq |K|-d-1$.
Hence $\chi(G')\leq |K'|+1\leq |K|-d=\chi(G)-d$. Note that the latter equality follows from our assumption that $(K,I)$ is maximal.
Now suppose that $k\geq |K|$. We contract $|K|$ arbitrary edges of $K$. The resulting graph $G'$ has chromatic number $\chi(G')=2\leq \chi(G)-d$. 
Hence, in both cases, we conclude that $(G,k)$ is a yes-instance.

\medskip
\noindent
Finally consider $\pi=\omega$. We use the previous result combined with the fact that split graphs are perfect and closed under edge contractions.
\qed
\end{proof}

In our next theorem we give two hardness results which, as explained in Section~\ref{s-intro}, show that Theorem~\ref{t-splitalpha} can be seen as best possible.
In their proofs we will reduce from the {\sc Red-Blue Dominating Set} problem. This problem takes as input a bipartite graph $G=(R\cup B,E)$ and an integer $k$, and asks whether there exists a {\em red-blue dominating set} of size at most~$k$, that is, a subset $D\subseteq B$ of at most $k$ vertices such that every vertex in $R$ has at least one neighbour in $D$.
This problem is \NP-complete, because it is equivalent to the \NP-complete problems {\sc Set Cover} and {\sc Hitting Set}~\cite{GJ79}.
The  {\sc Red-Blue Dominating Set} problem is  also $\W[1]$-complete when parameterized by~$|B|-k$~\cite{GJY11}.
Belmonte et al.~\cite{BGHP13} reduced from the same problem for showing that $1$-{\sc Contraction Blocker$(\Delta$)} is \NP-complete and \W[2]-hard (with parameter $k$)
for split graphs, but the arguments we use to prove  our results are quite different from the ones they used. 

\begin{theorem}\label{t-splitalpha2}
For $\pi\in \{\alpha,\chi,\omega\}$,  the {\sc Contraction Blocker$(\pi)$} problem, restricted to split graphs, is \NP-complete 
 as well as $\W[1]$-hard when parameterized by $d$.
\end{theorem}

\begin{proof}
The problem is in \NP\ for $\pi \in  \{\alpha,\chi,\omega\}$, as split graphs are closed under edge contraction and the three problems {\sc Clique}
{\sc Coloring} and {\sc Independent Set} are readily seen to be polynomial-time solvable on 
split graphs; hence, we can take the sequence of edge contractions as 
the certificate.
Recall that we reduce from {\sc Red Blue Dominating Set} in order to show \NP-hardness and \W[1]-hardness with parameter~$d$.

First consider $\pi=\alpha$.  
Let $G=(R\cup B,E)$ be a bipartite graph that together with an integer $k$ forms an instance of  {\sc Red-Blue Dominating Set}. 
We may assume without loss of generality that $k\leq |B|$. Moreover, we may assume that every vertex of~$R$ is adjacent to at least one vertex of $B$.
We add all possible edges between vertices in $R$. This yields a split graph $G^*$ with a split partition~$(R,B)$.
Because every vertex in $R$ is assumed to be adjacent to at least one vertex of $B$ in $G$, we find that $(R,B)$ is a minimal split partition of $G^*$.

Because  {\sc Red-Blue Dominating Set} problem is \NP-complete~\cite{GJ79} and \W[1]-complete when parameterized by $|B|-k$~\cite{GJY11},
it suffices to prove that $G$ has a red-blue dominating set of size at most~$k$ if and only if $(G^*,|B|-k)$ is a yes-instance of {\sc $(|B|-k)$-Contraction Blocker($\alpha$)}. We prove this claim below.   

First suppose that $G$ has a red-blue dominating set $D$ of size at most~$k$. Because $k\leq |B|$, we may assume without loss of generality that $|D|=k$ (otherwise we would just add some vertices from $B\setminus D$ to $D$).

In $G^*$ we contract every $u\in B\setminus D$ onto a neighbour in~$R$. In this way we $(|B|-k)$-contracted~$G^*$ into a graph $G'$. 
Note that $G'$ is a split graph that has a split partition $(R,D)$. Because 
every vertex in $R$ is adjacent to at least one vertex of $D$ in $G$ by definition of $D$, it is adjacent to at least one vertex of $D$ in $G^*$. 
The latter statement is still true for $G'$, as contracting an edge incident to a vertex $u\in B$ is equivalent to deleting $u$.
Hence, $(R,D)$ is a minimal split partition of $G'$, so $\alpha(G')=|D|$. Because $(R,B)$ is a minimal split partition of $G^*$, we have $\alpha(G^*)=|B|$.
This means that $\alpha(G')=|D|=|B|-(|B|-|D|)=\alpha(G^*)-(|B|-k)$.
We conclude that $(G^*,|B|-k)$ is a yes-instance of {\sc $(|B|-k)$-Contraction Blocker($\alpha$)}.

Now suppose that $(G^*,|B|-k)$ is a yes-instance of {\sc $(|B|-k)$-Blocker($\alpha$)}, that is, $G^*$ can be $(|B|-k)$-contracted into a graph $G'$ such that $\alpha(G')\leq \alpha(G^*)-(|B|-k)$.
Recall that $\alpha(G^*)=|B|$. Hence, $\alpha(G')\leq k$. Let $p$ be the number of contractions of edges with one end-vertex in $B$.
Note that any such contraction decreases the size of the independent set $B$ by exactly one.
If $p<|B|-k$, then $G'$ contains an independent set of size $|B|-p>k$, which would mean that $\alpha(G')>k$, a contradiction.
Hence, $p\geq |B|-k$, which implies that $p=|B|-k$ as we performed no more than $|B|-k$ contractions in total. 
Let $D$ denote the independent set obtained from $B$ after all edge contractions. Then we find that 
$k= |B|-(|B|-k)=|B|-p=|D|\leq \alpha(G')
\leq \alpha(G^*)-(|B|-k)
=|B|-(|B|-k)
=k$.
Hence, $|D|=\alpha(G')$, which means that $(D,R)$ is a minimal split partition of~$G'$.
This means that every vertex of $R$ is adjacent to at least one vertex of $D$ in~$G'$. Because all our contractions were performed on edges with one end-vertex in~$B$, we have only removed vertices from $G^*$, that is, $G'$ is an induced subgraph of $G^*$. Hence, every vertex of $R$ is adjacent to at least one vertex of $D$ in~$G'$.
Consequently, $D$ is a red-blue dominating set of $G$ with size $|D|=k$.

\medskip
\noindent
Now consider $\pi=\chi$.
Let $G=(R\cup B,E)$ be a bipartite graph that together with an integer $k$ forms an instance of  {\sc Red-Blue Dominating Set}. 
We may assume without loss of generality that $k\leq |B|$. Moreover, we may assume that every vertex of~$R$ is adjacent to at least one vertex of $B$.

We take the bipartite complement of $G$, that is, we construct the bipartite graph with partition classes $R$ and $B$, and we add an edge between any two vertices $u\in R$ and $v\in B$ if and only if $uv\notin E$. Then, we add all possible edges between vertices in $B$. Finally we add a new vertex
$x$ to the graph. We make $x$  adjacent to all vertices of $B\cup R$.
This yields a split graph $G^*$ with a split partition $(B\cup \{x\},R)$.
Because every vertex in $R$ is assumed to be adjacent to at least one vertex of $B$ in $G$, it  is non-adjacent to at least one vertex of $B$ in $G^*$. 
Hence, $(B\cup\{x\},R)$ is a maximal split partition of $G^*$ (we will explain the role of vertex $x$ in our construction later).
Similarly to the previous case, we claim that $G$ has a red-blue dominating set of size at most~$k$ if and only if $(G^*,|B|-k)$ is a yes-instance of {\sc $(|B|-k)$-Contraction Blocker($\chi$)}. We prove this claim below.

First suppose that $G$ has a red-blue dominating set $D$ of size at most~$k$. Because $k\leq |B|$, we may assume without loss of generality that $|D|=k$ (otherwise we would just add some vertices from $B\setminus D$ to $D$).

In $G^*$ we contract every $u\in B\setminus D$ onto $x$. In this way we $(|B|-k)$-contracted~$G^*$ into a graph
$G'$. Note that $G'$ is a split graph that has a split partition $(D\cup \{x\},R)$. Because every vertex in $R$ is adjacent to at least one vertex of $D$ in $G$ by definition of $D$, it is non-adjacent to at least one vertex of $D$ in $G^*$. The latter statement is still true for $G'$, as no vertex of $D\cup R$ was involved in any of the edge contractions performed.
Hence, $(D\cup \{x\},R)$ is a maximal split partition of $G'$, so $\chi(G')=|D|+1$. Because $(B\cup \{x\},R)$ is a maximal split partition of $G^*$, we have $\chi(G^*)=|B|+1$.
This means that $\chi(G')=|D|+1=k+1=|B|+1+k+1-(|B|+1)=\chi(G^*)-(|B|-k)$.
We conclude that $(G^*,|B|-k)$ is a yes-instance of {\sc $(|B|-k)$-Contraction Blocker($\chi$)}.

Now suppose that $(G^*,|B|-k)$ is a yes-instance of {\sc $(|B|-k)$-Blocker($\chi$)}, that is, $G^*$ can be $(|B|-k)$-contracted to a graph $G'$ such that $\chi(G')\leq \chi(G^*)-(|B|-k)$.
Recall that $\chi(G^*)=|B|+1$. Hence, $\chi(G')\leq k+1$. Let $p$ be the number of contractions of edges between two vertices of $B\cup \{x\}$. 
Note that any such contraction decreases the size of the clique $B\cup \{x\}$ by exactly one.
If $p<|B|-k$, then $G'$ contains a clique of size $|B|+1-p>k+1$, which would mean that $\chi(G')>k+1$, a contradiction.
Hence, $p\geq |B|-k$, which implies that $p=|B|-k$ as we performed no more than $|B|-k$ contractions in total. 
Let $B'$ denote the clique obtained from $B\cup \{x\}$ after all edge contractions. Then we find that 
$$
\begin{array}{lcl}
k+1 &= &|B|+1-(|B|-k)\\[1mm]
&=&|B|+1-p\\[1mm]
&=&|B'|\\[1mm]
&\leq &\chi(G')\\[1mm]
&\leq &\chi(G^*)-(|B|-k)\\[1mm]
&=&|B|+1-(|B|-k)\\[1mm]
&=&k+1.
\end{array}
$$
Hence, $|B'|=\chi(G')$, which means that $(B',R)$ is a maximal split partition of~$G'$.
This means that no vertex of $R$ is adjacent to all vertices of $B'$ in $G'$. 

We may assume without loss of generality that $x\in B'$, as we can view any edge contraction of an edge between a vertex $u\in B$ and $x$ as a contraction of $u$ onto $x$.
Furthermore, suppose we performed a contraction of an edge $uu'$ with $u,u'\in B$, say we contracted $u$ onto $u'$. We change this by contracting $u$ onto $x$ instead.
Because $x$ is adjacent to all vertices of $B\cup R$ in $G$, we find that $x$ is adjacent to all vertices (except to itself) of $G'$ and of any intermediate graph that we obtained while contracting $G$ into $G'$.
Hence, contracting $u$ onto $x$ is equivalent to deleting $u$. As such, contracting $u$ onto $x$ does not lead to a vertex $v\in R$ becoming adjacent to all vertices of $B'$. 
Consequently, the size of a maximum clique in the modified graph is also equal to $|B'|=\chi(G')$. As we can do the same for any other contraction of an edge between two vertices in~$B$,
we may assume without loss of generality that every edge contraction is a contraction of a vertex of~$B$ onto~$x$. 

Let $D=B'\setminus \{x\}\subseteq B$.
As noted, contracting a vertex of~$B$ onto $x$ is the same as deleting such a vertex of~$B$ from the graph.
Hence, every vertex of~$D$ has exactly the same neighbours in $G'$ as it has in $G^*$.
Because every vertex in $R$ is adjacent to $x$ but not to
all vertices of $B'=D\cup \{x\}$, we find that every vertex in $R$ is non-adjacent to at least one vertex of~$D$ in $G'$, and consequently, in $G^*$. 
Because $x\in B'$ and $|B'|=k+1$, we find that $|D|=k$.
We conclude that $D$ is a 
red-blue dominating set of $G$ with size $|D|=k$.

\medskip
\noindent
Finally, consider $\pi=\omega$. As split graphs are perfect and closed under edge contractions, this case follows directly from the previous case where $\pi=\chi$.\qed
\end{proof}

Regarding the {\sc Deletion Blocker$(\pi)$} problem, for $\pi\in \{\alpha,\chi,\omega\}$,  we know from \cite{CWP11} that it is \NP-complete. In the same paper it was shown that if $d$ is fixed, all three problems become polynomially solvable.

\section{Interval Graphs}\label{sec:interval}

Let $G=(V,E)$ be an interval graph with $n$ vertices and $m$ edges that corresponds to a set of intervals $\mathcal{I}=\{I_1, I_2, \ldots, I_n\}$ on the real line. Let $V=\{v_1,\ldots,v_n\}$ be such that vertex $v_i$
corresponds to interval $I_i$ for $i=1,\ldots,n$. Note that the class of interval graphs is closed under edge contraction. Indeed, contracting an edge $v_iv_i$ corresponds to removing the intervals $I_i$ and $I_j$ and adding a new interval $I_{ij}= I_i \cup I_j$.
It is well known (see e.g.~\cite{FG65}) that $G$ has at most $n$ maximal cliques which can be linearly ordered in $O(n+m)$ time so that
the maximal cliques containing a vertex $v_i$ appear consecutively for $i=1,\ldots,n$.

We first prove a useful lemma for the class of $C_4$-free graphs, which contains the class of interval graphs as a proper subclass.

\begin{lemma}
\label{lem:interval1}
Let $G=(V,E)$ be a $C_4$-free graph and let $v_1v_2\in E$. Let $G|v_1v_2$ be the graph obtained after the contraction of $v_1v_2$ and let $v_{12}$ be the new vertex replacing $v_1$ and $v_2$. Then 
every maximal clique $K$ in $G|v_1v_2$ containing $v_{12}$ corresponds to a maximal clique $K'$ in $G$ and vice versa, such that
\begin{enumerate}
\item[(a)] either $|K|=|K'|$ and $K\setminus\{v_{12}\}=K'\setminus\{v_1\}$;
\item[(b)] or $|K|=|K'|$ and $K\setminus\{v_{12}\}=K'\setminus\{v_2\}$;
\item[(c)] or $|K|=|K'|-1$ and $K\setminus\{v_{12}\}=K'\setminus\{v_1,v_2\}$.
\end{enumerate}
Moreover, every other maximal clique in $G|v_1v_2$ is a maximal clique in $G$ and vice versa.
\end{lemma}

\begin{proof}
Let $A_1$ (resp. $A_2$) be the set of neighbours of $v_1$ (resp. $v_2$) that are nonadjacent to $v_2$ (resp. $v_1$). Let $A_3$ be the set of vertices adjacent to both $v_1$ and $v_2$.  Now consider a clique~$K$ in $G|v_1v_2$ containing $v_{12}$. As $G$ is $C_4$-free, we find that $G$, and hence $G|v_1v_2$, contains no edge
between a vertex in $A_1$ and a vertex in $A_2$. Therefore we are in exactly one of the following cases: 
\begin{itemize}
\item [(i)] $K$ contains one or more vertices from both $A_1$ and $A_3$ but no vertices from $A_2$; 
\item [(ii)] $K$ contains one or more vertices from both $A_2$ and $A_3$ but no vertices from $A_1$;
\item [(iii)] $K$ contains one or more vertices from $A_1$ but no vertices from $A_2$ and $A_3$; 
\item [(iv)] $K$ contains one or more vertices from $A_2$ but no vertices from $A_1$ and $A_3$; 
\item [(v)] $K$ contains one or more vertices from $A_3$ but no vertices from $A_1$ and $A_2$. 
\end{itemize}

Suppose we are in case (i). 
Since $K$ is maximal, it follows that $(K\setminus \{v_{12}\})\cup \{v_1\}$ is a maximal clique in $G$ and thus outcome (a) holds. 
By symmetry, if we are in case (ii), outcome (b) holds.
Assume now that case (iii) occurs.
Since $K$ is maximal, it follows that $(K\setminus \{v_{12}\})\cup \{v_1\}$ is a maximal clique in $G$ and thus outcome (a) holds. By symmetry, we conclude that if case (iv) occurs, outcome (b) holds.
Finally, suppose that we are in case (v).
Then $(K\setminus \{v_{12}\})\cup \{v_1,v_2\}$ is a maximal clique in~$G$ and thus outcome (c) holds.
\qed
\end{proof}

Lemma \ref{lem:interval1} tells us that if we contract an edge $e$ in a $C_4$-free graph, every maximal clique containing both end-vertices of $e$ will have its size reduced by exactly one in the resulting graph, and moreover, the size of every other maximal clique of the original graph will remain the same and we do not create any new maximal clique.

\begin{lemma}
\label{lem:interval2}
Let $G=(V,E)$ be an interval graph and let $d\geq 0$ be an integer. Let 
$K^1$ be the first maximal clique of size strictly greater than $\omega(G)-d$ starting left on the real line, and let $I_x,I_y$ be the intervals with the rightmost right endpoints among all intervals corresponding to the vertices in $K^1$. Let $B\subseteq E$ be a set of edges such that the graph $G'$ obtained from $G$ after having contracted all edges from $B$ satisfies $\omega(G')\leq \omega(G)-d$. Then there exists a set $B'\subseteq E$ such that $B'=(B\setminus \{v_1v_2\})\cup \{xy\}$, where $v_1,v_2\in K^1$ and such that the graph $G''$ obtained from $G$ after contracting all edges in $B'$ satisfies $\omega(G'')\leq \omega(G)-d$. 
\end{lemma}

\begin{proof}
We first note that,
by their definition, $x$ and $y$ are
contained in all maximal cliques of size strictly greater than $\omega(G)-d$
that contain at least two vertices of  $K^1$.
Moreover, contracting the edge $xy$ instead of another edge $v_1v_2$ of $K^1$ does not create cliques of larger size, due to Lemma~\ref{lem:interval1}.
\qed
\end{proof}

Lemma \ref{lem:interval2} tells us that if for an interval graph the answer of  the {\sc Contraction Blocker($\omega$)} problem is {\tt yes}, then there always exists a set $B\subseteq E$ with $|B|\leq k$ such that $\omega(H)\leq \omega(G)-d$, where $H$ is the graph obtained from $G$ by contracting the edges of $B$, and $xy\in B$ where $x,y$ belong to the first maximal clique $K$ in $G$
with size strictly greater than $\omega(G)-d$ starting left on the real line and such that $I_x,I_y$ have the rightmost right endpoints among all intervals corresponding to vertices in $K$. Since interval graphs are closed under edge contractions, we can use this property recursively to obtain a polynomial-time algorithm for {\sc Contraction Blocker($\pi$)}, with $\pi\in\{\chi,\omega\}$, in interval graphs.

\begin{theorem}\label{thm:interval}
Let $\pi\in \{\chi,\omega\}$. Then {\sc Contraction Blocker($\pi$)} can be solved in polynomial time on interval graphs.
\end{theorem}

\begin{proof}
Since interval graphs are perfect and closed under edge contractions, we may assume without loss of generality that $\pi=\omega$. Let $G=(V,E)$ be an interval graph and let $d\geq 0$ be an integer. Our algorithm goes as follows. Let $K^1$ be the first maximal clique of size strictly greater than $\omega(G)-d$ starting left on the real line. By Lemma \ref{lem:interval2}, we know that if there exists a solution, then there exists one in which we contract the edge $xy$ where $x,y\in K^1$ are such that the corresponding intervals $I_x,I_y$ have the rightmost right endpoints among all intervals corresponding to vertices in $K^1$. So we contract the edge $xy$. Since the resulting graph is still an interval graph, we may repeat our procedure. We consider again the first maximal clique of size strictly greater than $\omega(G)-d$ starting left on the real line and contract the edge whose end-vertices correspond to the intervals with the rightmost right endpoints among all intervals corresponding to vertices in that clique. We continue like this until there is no more maximal clique of size strictly greater than $\omega(G)-d$ in the graph.

The correctness of our algorithm follows from Lemmas \ref{lem:interval1} and \ref{lem:interval2}. Indeed, by Lemma \ref{lem:interval1} we know that our choice of the edges that we contract is such that at each step there is at least one maximal clique of size strictly greater than $\omega(G)-d$ whose size is reduced by one and furthermore, we do not create any new maximal clique. Since an interval graph on $n$ vertices contains at most $n$ maximal cliques, it follows that our algorithms stop after at most $nd$ steps.  Since all maximal cliques of an interval graph can be found in time $O(n+m)$, where $m$ is the number of edges, we then find that our algorithm runs in time $O(nd(n+m))$. Finally, Lemma \ref{lem:interval2} ensures that the set of edges we choose to contract has minimum size. 
\qed
\end{proof}

The proof of Theorem~\ref{thm:interval} can be readily adapted to show polynomial-time solvability of the  {\sc Deletion Blocker($\pi$)} problem on interval graphs for $\pi\in \{\chi,\omega\}$.

\begin{theorem}
\label{thm:interval2}
Let $\pi\in \{\chi,\omega\}$. Then {\sc Deletion Blocker($\pi$)} can be solved in polynomial time on interval graphs.
\end{theorem}

We recall that for $\pi=\alpha$ the complexity of both problems is open for interval graphs.

\section{Chordal Graphs}\label{sec:chordal}

The following result shows that Theorem~\ref{thm:interval} cannot be generalized to chordal graphs.

\begin{theorem}\label{t-chordal}
For $\pi\in \{\chi,\omega\}$,  {\sc $1$-Contraction Blocker($\pi$)} is \NP-complete for chordal graphs.
\end{theorem}

\begin{proof}
Since chordal graphs are perfect 
and closed under taking edge contractions,
we may assume without loss of generality that $\pi=\omega$. 
As {\sc Clique} is polynomial-time solvable on chordal graphs, this means that the problem is in \NP\ (take the sequence of edge contractions as the certificate).
We reduce from {\sc Vertex Cover}, which is well known to be \NP-complete (see~\cite{GJ79}). 

Let $G=(V,E)$ be a graph that together with an integer $k$ forms an instance of {\sc Vertex Cover}. 
From $G$ we construct a chordal graph $G'$ as follows. We introduce a new vertex $y$ not in $G$. We represent each edge $e$ of $G$ by a clique $K_e$ in $G'$ of size $|V|$ so that $K_e\cap K_f=\emptyset$ whenever $e\neq f$. We represent each vertex $v$ of $G$ by a vertex in $G'$ that we also denote by $v$. Then we let the vertex set of $G'$ be $V\cup \bigcup_{e \in E}K_e \cup \{y\}$. We add an edge between every vertex in $K_e$ and a vertex $v\in V$ if and only if $v$ is incident with $e$ in $G$. In $G'$ we let the vertices of $V$ form a clique. Finally, we add all edges between $y$ and any vertex in $V\cup \bigcup_{e \in E}K_e$. Note that the resulting graph $G'$ is indeed chordal. Note also that $\omega(G')=|V|+3$ (every maximum clique consists of $y$, the vertices of a clique $K_e$ and their two neighbours in $V$).

We claim that $G$ has a vertex cover of size at most~$k$ if and only if $G'$ can be $k$-contracted to a graph $H$ with $\omega(H)\leq \omega(G')-1$.
First suppose that $G$ has a vertex cover $U$ of size at most~$k$. For each vertex $v\in U$, we contract the corresponding vertex $v$ in $G'$ to $y$. As $|U|\leq k$, this means that we
$k$-contracted $G'$ into a graph $H$. Since $U$ is a vertex cover, we obtain $\omega(H)\leq |V|+2=\omega(G')-1$.

Now suppose that $G'$ can be $k$-contracted to a graph $H$ with $\omega(H)\leq \omega(G')-1$. Let $S$ be a corresponding sequence of edge contractions (so $|S|\leq k$ holds).
By Lemma~\ref{lem:interval1} and the fact that chordal graphs are closed under taking edge contractions,  we find that no contraction in $S$ results in a new maximum clique. Hence, as we need to reduce the size of each maximum clique $K_{uv}\cup \{u,v,y\}$ by at least~1, we may assume without loss of generality that each contraction in $S$ concerns an edge with both its end-vertices in $V\cup \{y\}$. We construct a set $U$ as follows. If $S$ contains the contraction of an edge $uy$ we select $u$. If $S$ contains the contraction of an edge $uv$, we select one of $u,v$ arbitrarily. Because each maximum clique $K_{uv}\cup \{u,v,y\}$ must be reduced, we find that $U\subseteq V$ is a vertex cover. By construction, $|U|\leq k$. This completes the proof.\qed
\end{proof}

Similar arguments as in the above proof can be readily used to prove the following result, which shows that Theorem~\ref{thm:interval2} cannot be generalized to chordal graphs.

\begin{theorem}\label{t-chordal2}
For $\pi\in \{\chi,\omega\}$,  {\sc $1$-Deletion Blocker($\pi$)} is \NP-complete for chordal graphs.
\end{theorem}

\section{Six Dichotomy Results and $C_4$-free Perfect Graphs with $\omega=3$}\label{sec:general}

In this section we first prove that for $\pi\in\{\alpha,\chi, \omega\}$ the contraction and deletion blocker problems become very quickly \NP-hard when we increase $\pi$, that is, we prove Theorem~\ref{thm:NPC}.

\medskip
\noindent
{\bf Theorem~\ref{thm:NPC} (restated).}
{\it The following six dichotomies hold:
\begin{itemize}
\item [{\bf (i)}] {\sc Contraction Blocker($\alpha$)}  is polynomial-time solvable for graphs with $\alpha=1$ and 
{\sc 1-Contraction Blocker($\alpha$)} is \NP-complete for graphs with $\alpha=2$;\\[-8pt]
\item[{\bf (ii)}] {\sc Contraction Blocker($\chi$)}  is polynomial-time solvable for graphs with $\chi=2$ and 
 {\sc $1$-Contraction Blocker($\chi$)} is \NP-complete for graphs with $\chi=3$;\\[-8pt]
\item  [{\bf (iii)}] {\sc Contraction Blocker($\omega$)} is polynomial-time solvable for graphs with $\omega=2$ and
{\sc $1$-Contraction Blocker($\omega$)} is \NP-complete for graphs with $\omega=3$;\\[-8pt]
\item  [{\bf (iv)}] {\sc Deletion Blocker($\alpha$)} is polynomial-time solvable for graphs with $\alpha=1$ and\\ 
{\sc $1$-Deletion Blocker($\alpha$)} is \NP-complete for graphs with $\alpha=2$;\\[-8pt]
\item [{\bf (v)}] {\sc Deletion Blocker($\chi$)} is polynomial-time solvable for graphs with $\chi=2$ and\\
{\sc $1$-Deletion Blocker($\chi$)} is \NP-complete for graphs with $\chi=3$;\\[-8pt]
\item [{\bf (vi)}] {\sc Deletion Blocker($\omega$)} is polynomial-time solvable for graphs with $\omega=1$ and\\
{\sc $1$-Deletion Blocker($\omega$)}  is \NP-complete for graphs with $\omega=2$;\end{itemize}}

\begin{proof}
All six problems are readily seen to be in~\NP\ for the above graph classes (it suffices to take the sequence of edge contractions or vertex deletions as a certificate). We prove each of the six statements separately.

\medskip
\noindent
{\bf (i)}  The problem is trivial if $\alpha=1$.
As cobipartite graphs have independence number at most~2, we can apply Theorem~\ref{t-firstco} to obtain \NP-completeness if $\alpha=2$.

\medskip
\noindent
{\bf (ii)} The problem is trivial if $\chi\leq 2$. We now consider the class of graphs with $\chi=3$.
Recall that the problem {\sc Bipartite Contraction} is to test whether a graph can be made bipartite by at most $k$ edge contractions. It is readily seen that {\sc $1$-Contraction Blocker($\chi$)}   and {\sc Bipartite Contraction} are equivalent for graphs of chromatic number~3. Heggernes, van 't Hof, Lokshtanov and Paul~\cite{HHLP} observed that {\sc Bipartite Contraction} is \NP-complete
by reducing from the \NP-complete problem {\sc Edge Bipartization}, which is that of testing whether a graph can be made bipartite by deleting at most $k$ edges. Given an instance $(G,k)$ of {\sc Edge Bipartization}, they obtain an instance $(G',k')$ of {\sc Bipartite Contraction} by replacing every edge in $G$ by a path of sufficiently large odd length. 
Note that the resulting graph $G'$ has chromatic number~3 (assign colour~1 to the vertices of $G$ and give the new vertices colours~2 and~3).

\medskip
\noindent
{\bf (iii)} The problem is trivial if $\omega\leq 2$. We now consider the class of graphs with $\omega=3$.
We use a polynomial reduction from the problem {\sc ONE-IN-3-SAT}, which is well known to be \NP-complete (see~\cite{GJ79}). This problem has as input a set $X=\{x_1,\ldots, x_n\}$ of $n$ boolean variables and a collection $C=\{c_1,\ldots,c_m\}$ of clauses over $X\cup \bar X$ such that $\vert c_i\vert=3$
for $i=1,\ldots,m$.  The question is whether there a truth assignment for $X$ such that each clause of $C$ contains exactly one true literal.

Let $I=(X,C)$ be an instance of {\sc ONE-IN-3-SAT}. We construct an instance $(G,n+m)$ of {\sc $1$-Contraction Blocker($\omega$)}, where $G$ is constructed as follows  (see Fig. \ref{cliqfig1} for an example): 

\begin{itemize}
\item  For each variable $x\in X$, introduce five vertices forming a triangle and a square sharing exactly one edge. This yields the gadget for the variable $x$,
where the two edges that do not belong to the square correspond to the two literals $x$ and $\bar x$.
\item For each clause $c_i\in C$, introduce three vertices forming a triangle $T_i$. This yields the gadget for the clause $c_i$,
where each edge corresponds to one of the three literals forming $c_i$.
\item For every edge of a triangle $T_i$ corresponding to a literal $\lambda$, link its two end-vertices by a matching to the two end-vertices of the edge corresponding to~$\lambda$ in the variable gadget.
\end{itemize}

Observe that $(G,n+m)$ can be obtained in polynomial time. Moreover, $\omega(G)=3$ and $G$ contains exactly $n+m$ disjoint triangles. Thus, in order to obtain a graph $G'$ from $G$ with  $\omega(G')=2$, we need to contract at least one edge from each of these triangles.
We claim that $I$ is a yes-instance of {\sc ONE-IN-3-SAT} if and only if $(G,n+m)$ is a yes-instance of {\sc 1-Contraction Blocker($\omega$)}.

\begin{figure}[ht!]
	\centering
		\includegraphics[scale=0.5,keepaspectratio=true]{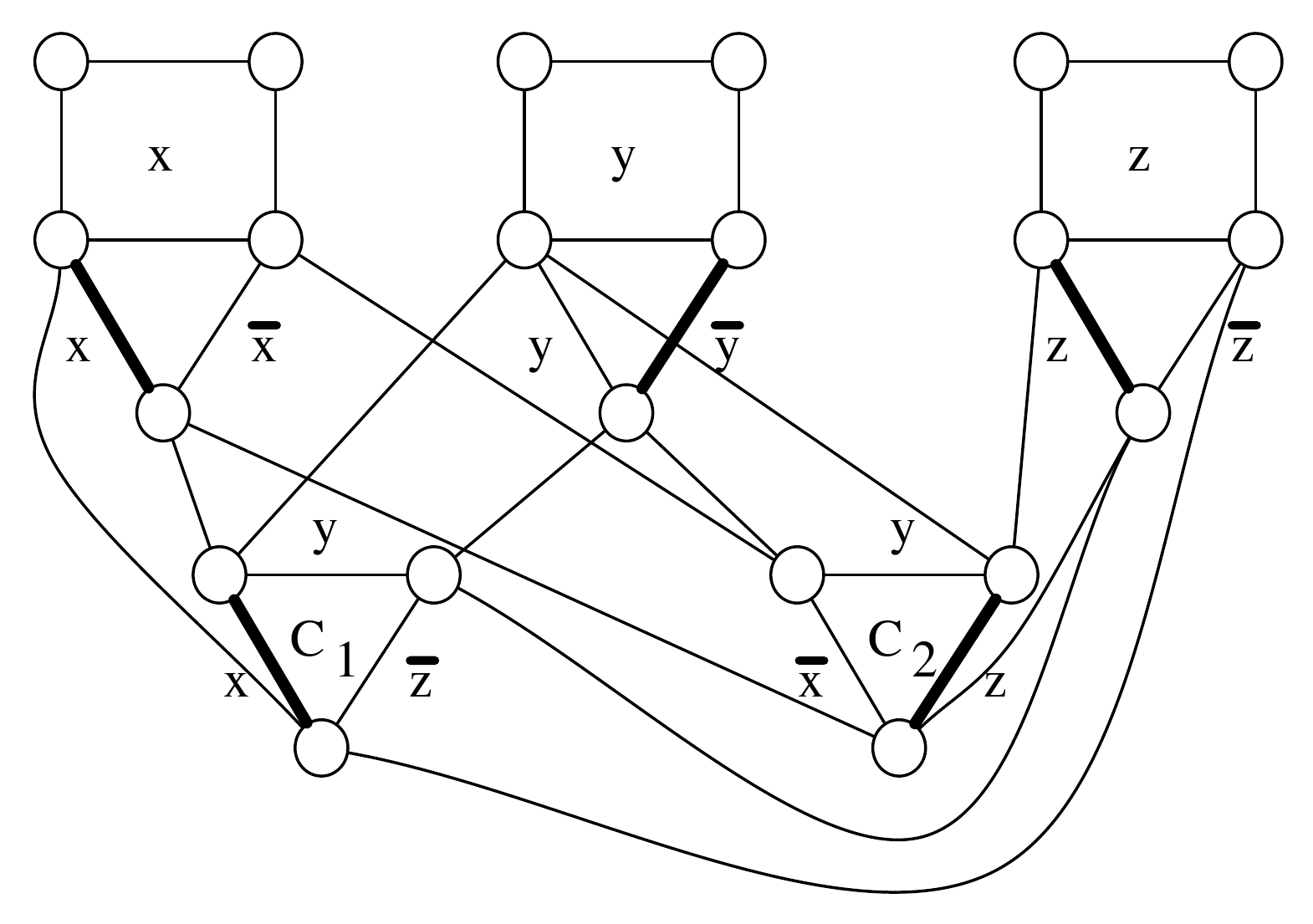}
	\caption{The graph $G$ encoding $I=(x,y,{\bar z})\vee({\bar x},y,z)$. Bold lines correspond to 
	 true literals.}
	\label{cliqfig1}
\end{figure}

First suppose that $I$ is a yes-instance. For each variable $x$ which is {\it true} (resp. {\it false}), we contract the edge corresponding to the literal $x$ (resp. the literal $\bar x$) in the triangle of the variable gadget; for each clause $c_i$, we contract the unique edge of the clause gadget corresponding to the literal which is set to {\it true} (see Fig. \ref{cliqfig2}). Thus we contract exactly $n+m$ edges, one in each of the $n+m$ disjoint triangles. For each clause gadget in $G$, the unique contracted edge is linked to the unique contracted edge in the variable gadget corresponding to the true literal. Hence the four original vertices are transformed into two adjacent vertices. 

We claim that no new triangles are created by performing the $n+m$ edge contractions. Indeed, when contracting an edge from a clause gadget, we do create a triangle $T$ one edge of which belongs to a variable gadget. But by construction, this edge will necessarily be contracted as well. Thus this triangle $T$ is transformed into a single edge. Hence $\omega(G')=2$, which means that $(G,n+m)$ is a yes-instance.

\begin{figure}[ht!]
	\centering
		\includegraphics[scale=0.5,keepaspectratio=true]{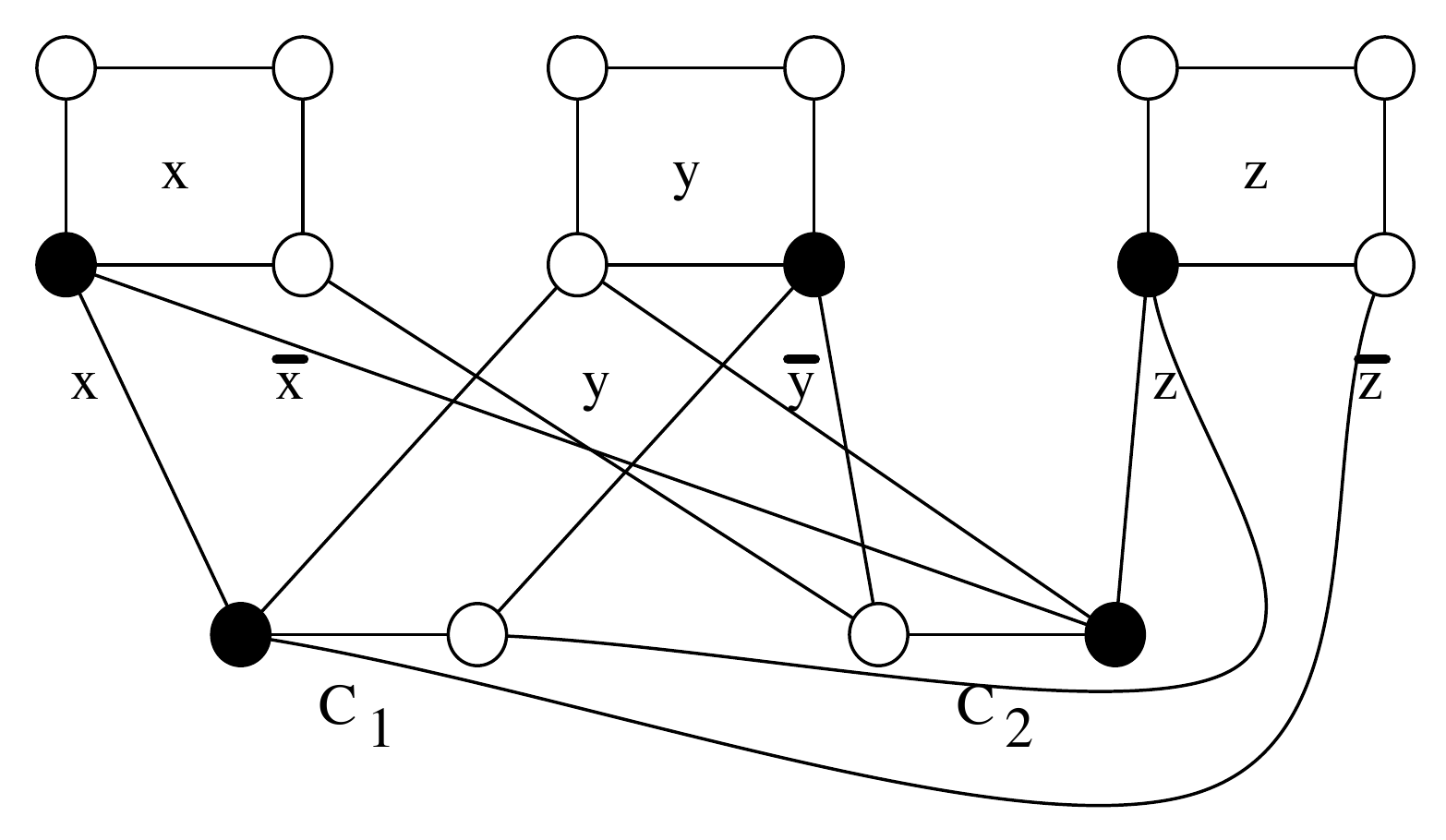}
	\caption{The graph $G$ after the contractions. The black vertices correspond to the contracted vertices.}
	\label{cliqfig2}
\end{figure}

Suppose now that $(G,n+m)$ is a yes-instance. This means that we can obtain a graph $G'$ with $\omega(G')=2$ by contracting $n+m$ edges of $G$.  Since~$G$ contains exactly $k$ disjoint triangles, we must, as already mentioned before, contract exactly one edge in each of these triangles. Furthermore, in a variable gadget we must contract an edge not belonging to the square, as otherwise a new triangle is created and hence we would need more than $n+m$ contractions, a contradiction. Let $e$ be an edge in a variable gadget that is contracted. Suppose that $e$ corresponds to a literal $\lambda$. In $G$, $e$ is contained in some squares containing edges of clause gadgets which correspond to $\lambda$. Thus, after this contraction, we create new triangles each containing an edge of a clause gadget corresponding to $\lambda$. It follows that we must contract the edges in the clause gadgets corresponding to the literal $\lambda$, otherwise triangles will remain in $G'$. Since we use $n+m$ edge contractions, exactly one edge in each clause gadget is contracted. Hence, by assigning the value true to the literal corresponding to the edge contracted in each variable gadget, one literal has value {\it true} and the other two have value {\it false} in each clause. This yields a positive answer for $I$, so $I$ is a yes-instance.

\medskip
\noindent
{\bf (iv) \& (vi)} Both problems are trivial if $\pi\in\{\alpha,\omega\}$ has value~1. Now consider the class of graphs with $\omega=2$, or equivalently the class of triangle-free graphs.
Since {\sc Vertex Cover} is \NP-complete for triangle-free graphs 
by Lemma~\ref{l-po2},
we conclude from Corollary~\ref{c-vc} that {\sc $1$-Deletion Blocker($\omega$)} is \NP-complete for triangle-free graphs. The remainder of statement~(iv) follows immediately after recalling that {\sc 1-Deletion Blocker($\alpha$)} can be solved by taking the complement of the input graph and solving {\sc 1-Deletion Blocker($\omega$)} instead.

\medskip
\noindent
{\bf (v)} First consider the class of graphs with $\chi=2$, which coincides with the class of bipartite graphs. Then the problem becomes equivalent to {\sc Independent Set}, which is
 polynomial-time solvable for bipartite graphs (due to K\"onig's Theorem; see, for example, \cite{Di05}).
Now consider the class of graphs with $\chi=3$.
Recall that the {\sc Maximum Induced Bipartite Subgraph} problem is to test if a given graph contains an induced bipartite subgraph with at least $k$ vertices for some integer~$k$ and that this problem is \NP-complete even for the class of 3-colourable perfect
graphs~\cite{AKKLR10}. As for 3-colourable graphs {\sc $1$-Deletion Blocker($\chi$)} is equivalent to {\sc Maximum Induced Bipartite Subgraph}, we find that {\sc $1$-Deletion Blocker($\chi$)} is \NP-complete for graphs with chromatic number~$3$. 

\medskip
\noindent
We have proven each of the six claims and thus have proven the theorem.\qed
\end{proof}

We note that the graph $G$ in the proof of Theorem~\ref{thm:NPC}~(iii) contains no induced {\it diamond} (the complete graph $K_4$ on four vertices minus an edge) and no induced {\it butterfly} (the graph with vertices $a,b,c,d,e$ and edges $ab,bc,ca,cd,de,ec$). As a graph $G$ is $K_4$-free if and only if $\omega(G)\leq 3$, we have in fact proven the following.

\begin{corollary}\label{c-butterfly}
The $1$-{\sc Contraction Blocker($\omega$)} problem is \NP-complete for the class of ($\mbox{butterfly,diamond},K_4$)-free graphs.
\end{corollary}

\medskip
\noindent
We use Theorem~\ref{thm:NPC}~(iii) to prove the following hardness result.

\begin{theorem}\label{t-perfect}
For $\pi\in \{\chi, \omega\}$,  $1$-{\sc Contraction Blocker($\pi$)}  is \NP-complete for the class of $C_4$-free perfect graphs with clique number~$3$.
\end{theorem}

\begin{proof}
As before, the problem is readily seen to be in \NP.
Let $\pi=\omega$, or equivalently, $\pi=\chi$. We adapt the construction used in the proof of Theorem~\ref{thm:NPC}~(iii)
by doing as follows for each edge $e$ of the graph $G$ in this proof. First we subdivide $e$. This gives us two new edges $e_1$ and $e_2$. We introduce two new non-adjacent vertices $u_e$ and $v_e$ and make them adjacent to both end-vertices of $e_1$. 
Denote the resulting graph by~$G^*$.
Note that we got rid of all the induced $C_4$s while not creating any new induced $C_4$ in this way. Hence $G^*$ is $C_4$-free. Moreover, we did not introduce any clique on four vertices. Hence, as $\omega(G)=3$, we also have $\omega(G^*)=3$.
The vertices of the original graph together with the subdivision vertices form a bipartite graph on top of which we placed a number of triangles. Hence, $G^*$ contains no odd hole and no odd antihole. By Theorem~\ref{t-spgt},
$G^*$ is perfect.

We increase the allowed number of edge contractions accordingly and observe that, because of the presence of the vertices $u_e$ and $v_e$ for each edge~$e$, we are always forced to contract the edge $e_1$, which gives us back the original construction extended with a number of pendant edges (which do not play a role). Note that we have left the class of $C_4$-free perfect graphs after contracting away the triangles, but this is allowed.
\qed
\end{proof}

We recall that {\sc Contraction Blocker($\alpha$)} is still open for the class of $C_4$-free perfect graphs as well as {\sc Deletion Blocker($\pi$)} for $\pi\in \{\alpha,\chi, \omega\}$, even if $d$ is fixed.

\section{$H$-free Graphs}\label{s-clas}

In this section we prove our complexity results for the six blocker problems restricted to $H$-free graphs, that is, we prove Theorem~\ref{t-mainmainmain}. To summarize, for $\pi\in \{\alpha,\omega,\chi\}$ we are able to give a dichotomy both for the contraction and deletion blocker problem except for one open case for the contraction blocker problem when $\pi=\omega$. 
We first consider $\pi=\alpha$, then $\pi=\omega$ and then $\pi=\chi$.

\subsection{When $\pi=\alpha$}\label{s-alphaclas}

We call a vertex {\it forced} if it is in every maximum independent set of a graph~\cite{CWP11}. 
Recall that the set of all forced vertices is called the {\it core} of a graph and that  
Boros, Golumbic and Levit~\cite{BGL02} proved that computing whether the core of a graph has size at least $k$ is co-\NP-hard for every fixed $k\geq 1$. 
As a special case of their result, the problem of testing the existence of a forced vertex is co-\NP-hard. 
We prove that the latter problem, or equivalently, {\sc Deletion Blocker($\alpha$)} with $d=k=1$, stays co-\NP-hard even for
 graphs of girth $p+1$, or equivalently,  $(C_3,\ldots,C_p)$-free graphs, for any constant $p\geq 3$ ((the {\it girth} of a graph is the length of a shortest cycle in it).

\begin{theorem}\label{t-po}
{\sc Deletion Blocker($\alpha$)} is co-\NP-hard for $(C_3,\ldots,C_p)$-free graphs for any constant~$p\geq 3$ even if $d=k=1$. 
\end{theorem}

\begin{proof}
Let $G$ be a graph. We pick one of its edges $uv$ and subdivide $uv$ twice, that is, we replace the edge $uv$ by two new vertices 
$x$ and $y$ and edges $ux$, $xy$, $yv$. We let $G'$ denote the resulting graph. Note that $\alpha(G')=\alpha(G)+1$ (see also~\cite{Po74}).
We claim that $G$ has a forced vertex if and only if $G'$ has a forced vertex.

First suppose that $G$ has a forced vertex~$s$. Then $s$ is also a forced vertex of~$G'$. In order to see this
consider a maximum independent set~$I'$ of $G'$.
For contradiction, suppose that $I'$ does not contain $s$. Recall that $I$ has size $\alpha(G)+1$. If $x$ is in~$I'$, then its neighbour $y$ is not in $I'$, and thus $I'\setminus \{x\}$ is a maximum independent set of $G$ that does not contain $s$, a contradiction. Hence
$x$ is not in $I'$, and for the same reason $y$ is not in $I'$ either. Then $u$ is in $I'$, as otherwise we could put 
$x$ in $I'$ to get a larger independent set than $I'$. However, we now find that $I'\setminus \{u\}$ is a maximum independent set of $G$ that does not contain~$s$, a contradiction.
Hence $s$ belongs $I'$. We conclude that $s$ is a forced vertex of $G'$ as well.

Now suppose that $G'$ has a forced vertex $s$. First suppose $s\in \{x,y\}$, say $s=x$. Then $v$ is a forced vertex of $G$. In order to see this consider a maximum independent set~$I$ of $G$. For contradiction, suppose that $I$  does not contain $v$. Then 
$I\cup \{y\}$ is a maximum independent
 set of $G'$ not containing $s=x$, a contradiction. Hence $s$ does not belong to $\{x,y\}$,
so $s$ must be in $G$. Then $s$ is also a forced vertex of $G$. In order to see this consider a maximum independent set~$I$ of~$G$. For contradiction, suppose that $I$ does not contain~$s$.
 As $u$ and $v$ are adjacent
 in $G$, not both of them are in $I$. Assume without loss of generality that
 $u$
 is not in $I$. Then $I\cup \{x\}$ is a maximum independent set of $G'$ that does not contain
 $s$,
 a contradiction. We conclude that $s$ is a forced vertex of $G$.
 
We now subdivide each edge of $G$ a sufficiently number of times (say $p$ times) so that the resulting graph $G''$ is $(C_3,\ldots,C_p)$-free. 
By repeatedly applying the above claim, we find that $G$ has a forced vertex if and only if $G''$ has a forced vertex.
As deciding whether a graph has a forced vertex is co-\NP-hard~\cite{BGL02}, the result follows.
 \qed
\end{proof}

Before we present our two complexity dichotomies for $\pi=\alpha$ we need one additional  observation.

\begin{lemma}\label{l-p44}
If $H$ is a $(3P_1,2P_2)$-free forest, then $H\ssi P_4$.
\end{lemma}

\begin{proof}
As $H$ is $3P_1$-free, $H$ contains at most two connected components. Suppose $H$ contains exactly two connected components. Then, as $H$ is $2P_2$-free, at least one of these components must be a $P_1$. As $H$ is $3P_1$-free, this means that $H$ is an induced subgraph of $P_1\oplus P_2$, so $H\ssi P_4$. Suppose $H$ is connected. As $H$ is $3P_1$-free, $H$ contains no claw and no path on more than five vertices. Hence, $H\ssi P_4$.\qed
\end{proof}

We are now ready to present our first dichotomy.

\begin{theorem}\label{t-indepnumber2}
Let $H$ be a graph.
 If $H\subseteq_i P_4$, then {\sc Deletion Blocker($\alpha$)} is polynomial-time solvable for $H$-free graphs, otherwise it is 
 \NP-hard or co-\NP-hard for $H$-free graphs.
\end{theorem}

\begin{proof}
Let $H$ be a graph. If $H\subseteq_i P_4$, then we use Corollary~\ref{c-coalpha} to obtain polynomial-time solvability.
Suppose $H$ is not an induced subgraph of $P_4$. 
If $H$ contains an induced cycle $C_r$ for some $r\geq 3$, then we pick $p=r+1$ and apply Theorem~\ref{t-po} to obtain
co-\NP-hardness even if $d=k=1$.
Note that for $r=5$, we could have applied Theorem~\ref{t-splitalpha2} to obtain \NP-hardness, as split graphs are $C_5$-free.
Similarly, if $r\geq 6$, then $H$ contains an induced $2P_2$ and we could have applied Theorem~\ref{t-splitalpha2} (as split graphs are $2P_2$-free) to obtain \NP-hardness as well.

Now assume that $H$ is forest. As $H$ is not an induced subgraph of $P_4$, by Lemma~\ref{l-p44}
either $2P_2\ssi H$ or $3P_1\ssi H$.
If $2P_2\ssi H$, then we apply Theorem~\ref{t-splitalpha2} again to obtain \NP-hardness.
If $3P_1\ssi H$, then we use Theorem~\ref{thm:NPC}~(iv) to obtain \NP-hardness even if $d=1$, after observing that
a graph~$G$ is $3P_1$-free if and only if $\alpha(G)=2$.
\qed
\end{proof}

\noindent
{\bf Remark 2.} Recall that $H$-free graphs are closed under vertex deletion. Hence, {\sc Deletion Blocker($\alpha$)} for $H$-free graphs will be in \NP\  if we can solve {\sc Independent Set} for $H$-free graphs in polynomial time; in that case we can take a sequence of vertex deletions as certificate. To give an example, {\sc Independent Set} is polynomial-time solvable for $P_5$-free graphs~\cite{LVV14}. Hence, for $P_5$-free graphs, {\sc Deletion Blocker($\alpha$)} is not only \NP-hard (which, as argued in the proof of Theorem~\ref{t-indepnumber2}, follows from
Theorem~\ref{t-splitalpha}) but even \NP-complete.
 
\medskip
\noindent
We now consider the edge contraction variant and present our second dichotomy.

\begin{theorem}\label{t-indepnumber0}
Let $H$ be a graph.
 If $H\subseteq_i P_4$, then {\sc Contraction Blocker($\alpha$)} is polynomial-time solvable for $H$-free graphs, otherwise it is 
\NP-hard for $H$-free graphs.
\end{theorem}

\begin{proof}
Let $H$ be a graph.
If $H$ is an induced subgraph of $P_4$, then we use Corollary~\ref{c-coalpha} to obtain polynomial-time solvability. 
Now suppose that $H$ is not an induced subgraph of $P_4$.
If $H$ contains an induced cycle that is odd, then we use Theorem~\ref{thm:bipartite} to obtain \NP-hardness. 
If $H$ contains an induced cycle that is even, then $H$ either contains an induced $C_4$ or, if the even cycle has at least six vertices, an induced $2P_2$. This means that we can use Theorem~\ref{t-splitalpha2} to obtain \NP-hardness after recalling that split graphs are $(2P_2,C_4)$-free.
Assume $H$ contains no cycle. Then $H$ is a forest.
If $H$ contains an induced $3P_1$, then we use Theorem~\ref{thm:NPC}~(i) to obtain \NP-hardness even if $d=1$, after observing that
a graph~$G$ is $3P_1$-free if and only if $\alpha(G)=2$.
Assume $H$ is $3P_1$-free. 
Then $2P_2\ssi H$ by Lemma~\ref{l-p44}, which means we can use  Theorem~\ref{t-splitalpha2} again to obtain \NP-hardness. 
\qed
\end{proof}

\subsection{When $\pi=\omega$}\label{s-omegaclas}

The complexity dichotomy for {\sc Deletion Blocker($\omega$)} follows immediately from Theorem~\ref{t-indepnumber2} after making two observations. First,
{\sc Deletion Blocker($\omega$)} for $H$-free graphs is equivalent to {\sc Deletion Blocker($\alpha$)} for $\overline{H}$-free graphs. Second, the graph $P_4$ is self-complementary, that is, $\overline{P_4}=P_4$.

\begin{theorem}
\label{thm-deletion-omega}
Let $H$ be a graph.
If $H\subseteq_i P_4$, then {\sc Deletion Blocker($\omega$)} is polynomial-time solvable for $H$-free graphs; otherwise it is co-\NP-hard or \NP-hard for $H$-free graphs.
\end{theorem}
We now consider the {\sc Contraction Blocker($\omega$)} problem for $H$-free graphs.
We start by giving a sufficient condition for computational hardness. 
Let ${\cal G}$ be a graph class with the following property: if $G\in {\cal G}$, then so are $2G$ and $G\oplus K_r$ for any $r\geq 1$. We call such a graph class {\it clique-proof}.

\begin{theorem}\label{t-cll}
If {\sc Clique}  is \NP-complete for a clique-proof graph class~${\cal G}$, then  {\sc Contraction Blocker($\omega$)} is co-\NP-hard for ${\cal G}$, even if $d=k=1$.
\end{theorem}

\begin{proof}
Let ${\cal G}$ be a graph class that is clique-proof. From a given graph $G\in {\cal G}$ and given integer $\ell\geq 1$ we construct the graph $G'=2G\oplus K_{\ell+1}$. Note that $G'\in {\cal G}$ by definition and that $\omega(G')=\max\{\omega(G),\ell+1\}$. 
It suffice to prove that $\omega(G)\leq \ell$ if and only if $G'$ can be $1$-contracted into a graph $G^*$ with $\omega(G^*)\leq \omega(G')-1$. 

First suppose that  $\omega(G)\leq \ell$. Then $\omega(G')=\omega(K_{\ell+1})=\ell+1$. In $G'$ we contract an edge of the $K_{\ell+1}$. This yields the graph $G^*=2G\oplus K_\ell$, which has
clique number $\omega(G^*)=\ell$, as $\omega(K_\ell)=\ell$ and $\omega(G)\leq \ell$. As $\omega(G')=\ell+1$, this means that  $\omega(G^*)\leq \omega(G')-1$. 

Now suppose that $G'$ can be $1$-contracted into a graph $G^*$ with $\omega(G^*)\leq \omega(G')-1$. 
As contracting an edge in one of the two copies of $G$ in $G'$ does not lower the clique number of $G'$, the contracted edge must be in the $K_{\ell+1}$, that is, $G^*=2G\oplus K_\ell$.
As this did result in a lower clique number, we conclude that $\omega(G')=\omega(K_{\ell+1})=\ell+1$ and $\omega(G^*)=\omega(2G\oplus K_\ell)=
\max\{\omega(G),\ell\}=\ell$.
The latter equality implies that $\omega(G)\leq \ell$.\qed
\end{proof}

We need a number of special graphs, namely the {\it cobanner}, {\it bull}, the aforementioned butterfly and the
 {\it paw} (the graph $\overline{P_1\oplus P_3}$), which are all displayed in Figure~\ref{fig:prelim}.
\begin{figure}
\begin{center}
\includegraphics[scale=0.6]{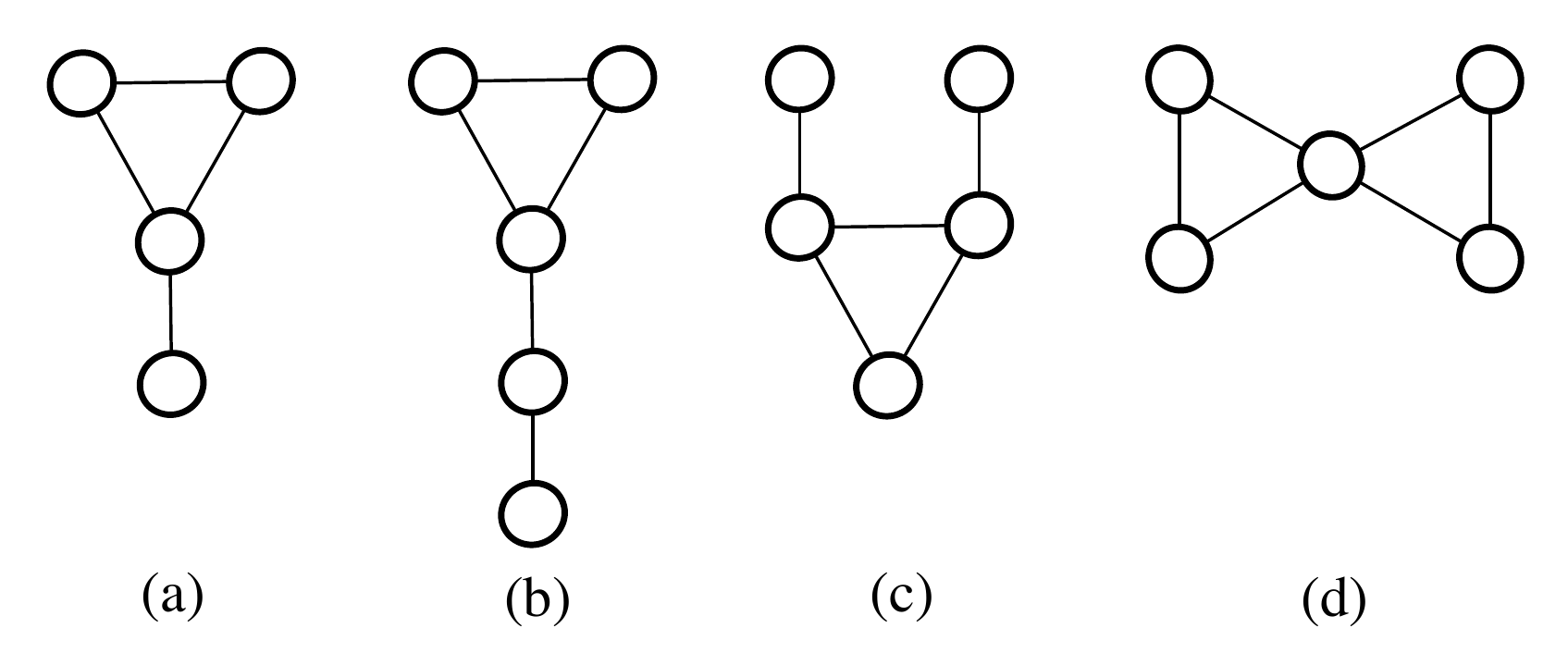}
\caption{(a) Paw. (b) Cobanner. (c) Bull. (d) Butterfly.}
\label{fig:prelim}
\end{center}
\end{figure}

We also need the following lemma from Poljak.

\begin{lemma}[\cite{Po74}]\label{lem:poljak2}
The {\sc Clique} problem is \NP-complete for the following classes:
$(C_5,P_5)$-free graphs, $K_{1,3}$-free graphs, $\mbox{cobanner}$-free graphs and $(\mbox{bull},P_5)$-free graphs.
\end{lemma}

We use Lemma~\ref{lem:poljak2} in the proof of our next lemma.

\begin{lemma}\label{t-cliquenumber}
Let $H$ be a connected graph. If $H$ is neither an induced subgraph of $P_4$ nor of 
the paw, then $1$-{\sc Contraction Blocker($\omega$)} is \NP-hard or co-\NP-hard for $H$-free graphs.
\end{lemma}

\begin{proof}
Let $H$ be a connected graph that is neither an induced subgraph of $P_4$ nor of the paw.
If $H$ contains an induced $C_4$, use Theorem~\ref{t-perfect}. If $H$ contains an induced $K_4$, diamond or butterfly, use 
Corollary~\ref{c-butterfly}.
If $H$ contains an induced $K_{1,3}$, $C_5$, $P_5$, bull or cobanner, use Lemma~\ref{lem:poljak2} with Theorem~\ref{t-cll}. So
from now on we may assume that $H$ is $(C_4,C_5,P_5,K_{1,3},K_4,\mbox{diamond},\mbox{bull},\mbox{butterfly},$ $\mbox{cobanner})$-free.
Below we show that this leads to a contradiction.

First suppose that $H$ contains no cycle. Then, as $H$ is connected, $H$ is a tree. Because $H$ is $K_{1,3}$-free, $H$ is a path. 
Our assumption that $H$ is neither an induced subgraph of $P_4$ nor of the paw
implies that $H$ contains an induced $P_5$, which is not possible as $H$ is $P_5$-free.

Now suppose that $H$ contains a cycle $C$. Then $C$ must have exactly three vertices, because $H$ is $(C_4,C_5,P_5)$-free.
As $H$ is not an induced subgraph of the paw, we find that $H$ contains at least one vertex~$x$ not on $C$.
As $H$ is connected, we may assume that $x$ has a neighbour on $C$. Because $H$ is $(\mbox{diamond},K_4)$-free, $x$ has exactly one neighbour on $C$. Let $v$ be this neighbour. Hence, $H$ contains an induced paw (consisting of $x$, $v$ and the other two vertices of $C$). 
As $H$ is not an induced subgraph of the paw
 and $H$ is connected, it follows that $H$ contains a vertex $y\notin V(C)\cup \{x\}$ that is adjacent to a vertex on $C$ or to~$x$. 

Suppose that $y$ is adjacent to a vertex of $C$. Then, as $H$ is $(\mbox{diamond},K_4)$-free, $y$ has exactly one neighbour $u$ in $C$. If $u=v$ then $H$ either contains an induced claw (if $x$ and $y$ are non-adjacent) or an induced butterfly (if $x$ and $y$ are adjacent). Since, by our assumption, this is not possible, it follows that $u\neq v$. Then, because $H$ is bull-free, we deduce that $x$ and $y$ are adjacent. However, then the vertices, $u,v,x,y$ form an induced $C_4$, which is not possible as $H$ is $C_4$-free. We conclude that $y$ is not adjacent to a vertex of $C$, so $y$ must be adjacent to $x$ only. However, then $H$ contains an induced cobanner, a contradiction. This completes the proof of Lemma~\ref{t-cliquenumber}.
\qed
\end{proof}

A graph~$G$ is \emph{complete multipartite} if~$V(G)$ can be partitioned into~$k$ independent sets $V_1,\ldots,V_k$ for some integer~$k$, such that two vertices are adjacent if and only if they belong to two different sets~$V_i$ and~$V_j$.
We need a result of Olariu on paw-free graphs. 

\begin{lemma}[\cite{Ol88}]
\label{l-paw}
Every connected paw-free graph is either triangle-free or complete multipartite.
\end{lemma}

We are ready to present our result for {\sc Contraction Blocker($\omega$)} restricted to $H$-free graphs. This is the only result where we do not have a dichotomy due to one missing case.

\begin{theorem}\label{t-cliquenumber2}
Let $H\neq C_3\oplus P_1$ be a graph. If $H\ssi P_4$ or $H\ssi \mbox{paw}$, then {\sc Contraction Blocker$(\omega$)} is polynomial-time solvable for $H$-free graphs, otherwise it is \NP-hard or co-\NP-hard for $H$-free graphs.
\end{theorem}

\begin{proof}
First assume that $H$ is connected.
If $H$ is an induced subgraph of $P_4$ then we use Corollary~\ref{c-coalpha}. If $H$ is an induced subgraph of the paw, then we know from
Lemma~\ref{l-paw} that $G$ is either $C_3$-free or complete multipartite. In the first case one must contract all the edges of an $H$-free graph in order to decrease its clique number. Hence {\sc Contraction Blocker($\omega$)} is polynomial-time solvable for $C_3$-free graphs. In the second case $H$ is $P_4$-free, so we can use Corollary~\ref{c-coalpha} again. If $H$ is  neither an induced subgraph of $P_4$ nor of the paw, then we use 
Lemma~\ref{t-cliquenumber}.

Now assume that $H$ is not connected. If $H$ contains a connected component that is not an induced subgraph of $P_4$ or the paw 
then we use Lemma~\ref{t-cliquenumber} again. Assume that each connected component of $H$ is an induced subgraph of $P_4$ or the paw. If $3P_1\ssi H$ or $2P_2\ssi H$ then we use Theorem~\ref{t-cohard} or Theorem~\ref{t-splitalpha2}, respectively. Hence, $H\in \{2P_1,P_2\oplus P_1,C_3\oplus P_1\}$. In the first two cases $H\ssi P_4$ and thus we can use Corollary~\ref{c-coalpha}, whereas we excluded the last case.
\qed
\end{proof}

\subsection{When $\pi=\chi$}\label{s-chiclas}

Recall that {\sc Deletion Blocker($\chi$)} and {\sc Contraction Blocker($\chi$)} are called {\sc Critical Vertex} and 
{\sc Contraction-Critical Edge}, respectively, if $d=k=1$. We need the following result announced in~\cite{PPR17b}; see~\cite{PPR17c} for its proof.

\begin{theorem}[\cite{PPR17b}]\label{t-critical}
If a graph $H\ssi  P_4$ or of $H\ssi P_1\oplus P_3$, then
{\sc Critical Vertex} and {\sc Contraction-Critical Edge} restricted to $H$-free graphs are polynomial-time solvable, otherwise they are \NP-hard or co-\NP-hard.
\end{theorem}

We also need the following result of Kr\'al', Kratochv\'{\i}l, Tuza, and Woeginger.

\begin{theorem}[\cite{KKTW01}]\label{t-dicho}
Let $H$ be a graph. If $H\ssi  P_4$ or of $H\ssi P_1\oplus P_3$, then
{\sc Coloring} is polynomial-time solvable for $H$-free graphs, otherwise it is \NP-complete for $H$-free graphs.
\end{theorem}

We also need the following lemma.

\begin{lemma} \label{l-known2}
{\sc Deletion Blocker($\chi$)} is polynomial-time solvable for $3P_1$-free graphs.
\end{lemma}

\begin{proof}
Let $G=(V,E)$ be a $3P_1$-free graph with $|V|=n$ and let $k \geq 1$ be an integer. Consider an instance $(G,k,d)$ of {\sc Deletion Blocker($\chi$)}. We proceed as follows.
Consider an optimal colouring of $G$.
 Since $G$ is $3P_1$-free, the size of each colour class is at most~$2$. Moreover, the number of colour classes of size 1 is the same for every optimal colouring of $G$. Let~$\ell$ be this number. Hence, there are $\frac{n-\ell}{2}$ colour classes of size~2 and $\chi(G)=\ell+\frac{n-\ell}{2}$. 
 
 Now $(G,k,d)$ is a yes-instance if and only if we can obtain a graph $G'$ from $G$ by deleting at most $k$ vertices such that $\chi(G') \leq \chi(G)-d=\ell+\frac{n-\ell}{2} -d$. Since $G'$ is also $3P_1$-free, the colour classes in any optimal colouring of $G'$ have size at most~2 and thus, $G'$ contains at most $2(\ell+\frac{n-\ell}{2} -d)=n+\ell-2d$ vertices. In other words, we need to delete at least $2d-\ell$ vertices from $G$ in order to get such a graph $G'$. As such, $(G,k,d)$ is a no-instance if $k < 2d-\ell$. 
 
 Next we will show that if $k \geq 2d-\ell$, then $(G,k,d)$ is a yes-instance and this will complete the proof. If $d \leq \ell$, we delete $d$ vertices representing colour classes of size~1. If $d > \ell$, we delete the $\ell$ vertices representing the colour classes of size 1 and $2(d-\ell)$ vertices of $d-\ell$ colour classes of size~2. In this way we obtain a graph $G'$ whose chromatic number is exactly $\chi(G)-d$.
 
 Due to the above, all we need to do is check if $k\geq 2d-\ell$. This can be done in polynomial time, since we can compute $\ell$ in polynomial time due to Theorem~\ref{t-dicho}.\qed
\end{proof}

Two disjoint subsets of vertices in a graph are {\it complete} if there is an edge between every vertex of $A$ and every vertex of $B$.
Lemma~\ref{l-paw} implies the following lemma, which we use together with Corollary~\ref{c-coalpha} and Lemma~\ref{l-known2} to prove Lemma~\ref{l-new}. 

\begin{lemma}\label{l-copaw}
The vertex set of every $(P_1\oplus P_3)$-free graph $G$ can be decomposed into two disjoint sets $A$ and $B$ such that
$G[A]$ is $3P_1$-free, $G[B]$ is $P_4$-free and $A$ and $B$ are complete to each other.
\end{lemma}

\begin{proof}
Let $G=(V,E)$ be a $(P_1\oplus P_3)$-free graph. Then $\overline{G}$ is $\overline{P_1\oplus P_3}$-free. By Lemma~\ref{l-paw}
every connected component of $\overline{G}$ is triangle-free or complete multipartite. Let $A$ be the union of the vertices 
of all triangle-free components. Then $\overline{G}[A]=K_3$-free, so $G[A]$ is $3P_1$-free. Let $B=V\setminus A$.
As every component of $\overline{G}[B]$ is complete multipartite, $\overline{G}[B]$ is $P_4$-free. As 
$\overline{P_4}=P_4$, this means that $G[B]$ is $P_4$-free.
Moreover, $A$ and $B$ are complete to each other in~$G$.\qed
\end{proof}

\begin{lemma}\label{l-new}
{\sc  Deletion Blocker($\chi$)}  is polynomial-time solvable for $(P_1\oplus P_3)$-free graphs.
\end{lemma}

\begin{proof}
Let $(G,d,k)$ be an instance of {\sc Vertex Deletion Blocker($\chi$)}, where $G=(V,E)$ is  $(P_1\oplus P_3)$-free.
By Lemma~\ref{l-copaw}, the vertex set of $G$ can be decomposed into two disjoint sets $A$ and $B$ such that
$G_1=G[A]$ is $3P_1$-free, $G_2=G[B]$ is $P_4$-free and $A$ and $B$ are complete to each other. The latter implies 
that $\chi(G)=\chi(G_1)+\chi(G_2)$. Moreover, this property is maintained when deleting vertices from $V$. 
For each pair $(k_1,k_2)$ with $k_1+k_2=k$ we check by how much we can decrease $\chi(G_1)$ using at most $k_1$ vertex deletions and by how much we can decrease $\chi(G_2)$ using at most $k_2$ vertex deletions. We can do this in polynomial time 
by Corollary~\ref{c-coalpha} and Lemma~\ref{l-known2}, respectively.
We keep track of the maximum sum of these values. In the end, we are left to check if the value found is at least $d$ or not.
Since the number of pairs $(k_1,k_2)$ is at most $k$, the total running time is polynomial.\qed
\end{proof}

We can now state and prove the following two dichotomies. 

\begin{theorem}\label{t-main2}
Let $H$ be a graph. Then the following holds:
\begin{itemize}
\item 
If $H\ssi P_1\oplus P_3$ or $P_4$, then {\sc Deletion Blocker$(\chi)$} for $H$-free graphs is polynomial-time solvable, and it is \NP-hard or co-\NP-hard otherwise.\\[-8pt]
\item 
If $H\ssi P_4$, then  {\sc Contraction Blocker$(\chi)$} for $H$-free graphs is polynomial-time solvable for $H$-free graphs, and it is
\NP-hard or co-\NP-hard otherwise.
\end{itemize}
\end{theorem}

\begin{proof}
Let $H$ be a graph. If $H$ is neither an induced subgraph of $P_4$ nor of $P_1\oplus P_3$, then for both problems we can apply 
Theorem~\ref{t-critical}. If $H\ssi P_4$, then for both problems we apply Corollary~\ref{c-coalpha}. In the remaining case
$H=3P_1$ or $H=P_1\oplus P_3$. Then applying Lemma~\ref{l-new} gives us the desired dichotomy for {\sc Deletion 
Blocker($\chi$)}. And applying Theorem~\ref{t-cohard} gives us the desired dichotomy for  {\sc Contraction Blocker($\chi$)} 
 after recalling that cobipartite graphs are $3P_1$-free.\qed
 \end{proof}
 
After proving Theorem~\ref{t-main2}  we have shown all six cases of Theorem~\ref{t-mainmainmain}.
 Note that, unlike the case $d=k=1$ (see Theorem~\ref{t-critical}), the complexity dichotomies of the problems {\sc Contraction Blocker($\chi$)} and {\sc Deletion Blocker($\chi$)} restricted to $H$-free graphs are different when $H$ is disconnected.

\section{Future Work}\label{s-con}

We aim to solve the blank entries in Table~\ref{t-thetable2}. In particular, we pose the following open problems:

\begin{itemize}
\item [Q1.] Determine the complexity of {\sc Contraction Blocker($\alpha$)} for interval graphs.\\[-8pt]
\item [Q2.] Determine the complexity of {\sc Deletion Blocker($\alpha$)} for interval graphs.
\end{itemize}
We observe that the complexity of the two problems in Q1 and Q2 is unknown for interval graphs even if $d$ is fixed. 
We also aim to research the complexity of {\sc $1$-Contraction Blocker($\alpha$)} for bipartite graphs and chordal graphs, and the complexity of 
{\sc $1$-Deletion Blocker($\alpha$)} for perfect graphs and chordal graphs.

In addition to the above it would be interesting to generalize our results for the blocker problems restricted to $H$-free graphs in Section~\ref{s-clas} to families of more than one forbidden induced subgraph~$H$. However, we still need to complete one stubborn remaining case for one problem:

\begin{itemize}
\item [Q3.] Determine the complexity of {\sc Contraction Blocker($\omega$)} for $(C_3\oplus P_1)$-free graphs.
\end{itemize}
 
We observe that it is not difficult to construct graph classes for which a blocker problem is tractable, but the original problem is \NP-complete. However, we do not know of such examples of hereditary graph classes. Hence
it would be interesting to solve the following question.

\begin{itemize}
\item [Q4.] For $\pi \in \{\alpha,\omega,\chi\}$, are {\sc Contraction Blocker($\pi$)} and {\sc Deletion Blocker($\pi$)} 
computationally hard on every hereditary graph class~${\cal G}$, for which {\sc Independent Set}, {\sc Clique} or {\sc Coloring}, respectively, is \NP-complete?
\end{itemize}

Several computationally hard cases of our dichotomies for $H$-free graphs in Theorem~\ref{t-mainmainmain} hold in fact even
when $d=1$ or $d=k=1$. In particular, from Theorems~\ref{t-critical} and~\ref{t-main2} we immediately deduce that
if $H\ssi P_1\oplus P_3$ or $P_4$, then {\sc 1-Deletion Blocker$(\chi)$} for $H$-free graphs is polynomial-time solvable, and \NP-hard or co-\NP-hard otherwise. However, for the other five variants we still have a number of missing cases to solve.

\noindent

Finally, we aim to determine a dichotomy with respect to $H$-free graphs for the variant ($\pi\in \{\alpha,\omega,\chi\}$), where $S$ consists of other graph operations, for instance when $S$ consists of an edge deletion. This variant has been less studied than the vertex deletion and edge contraction variant. The reason for this is that no class of $H$-free graphs is closed under edge deletion, whereas such a class is closed under vertex deletion, and in the case when $H$ is a linear forest, under edge contraction as well. For $\pi=\chi$ we are close to a dichotomy.
Recall that an edge of a graph is critical or contraction-critical if its deletion or contraction, respectively, reduces the chromatic number of~$G$ by~1. It is known that an edge is contraction-critical if and only if it is critical~\cite{PPR17b}.
Hence by Theorem~\ref{t-critical} we only need to consider the cases where $H\ssi  P_4$ or $H\ssi P_1\oplus P_3$.
Bazgan et al.~\cite{BBPR} showed that the edge deletion variant for chromatic number is polynomial-time solvable for threshold graphs, that is,
for  $(C_4,2P_2,P_4)$-free graphs, and \NP-hard for cobipartite graphs, and thus for $3P_1$-free graphs. 
This means that the only two open cases for chromatic number are when $H=P_1+P_2$ or $H=P_4$.

\end{document}